\documentclass[a4paper,fleqn]{llncs}

\usepackage[utf8]{inputenc}
\usepackage[T1]{fontenc}
\usepackage{amsmath}
\usepackage{amssymb}
\usepackage{xspace}
\usepackage{multicol}
\usepackage{stackengine}
\usepackage{pgf}
\usepackage{tikz}
\usetikzlibrary{arrows,automata}
\usepackage{booktabs}

\newcommand{\cceq}{\mathop{::=}}

\renewcommand{\epsilon}{\varepsilon}
\renewcommand{\phi}{\varphi}

\newcommand{\pow}[1]{2^{#1}}
\newcommand{\nats}{\mathbb{N}}
\newcommand{\size}[1]{|#1|}
\newcommand{\set}[1]{\{ #1 \}}
\newcommand{\ap}[0]{\mathrm{AP}}

\newcommand{\F}{{\mathbf{F\,}}}
\newcommand{\G}{{\mathbf{G\,}}}
\newcommand{\U}{{\mathbf{\,U\,}}}
\newcommand{\X}{{\mathbf{X\,}}}

\newcommand{\K}{{\mathcal{K}}}

\newcommand{\ltl}{{LTL}\xspace}
\newcommand{\hyltl}{{HyperLTL}\xspace}
\newcommand{\hyctlstar}{{HyperCTL$^*$}\xspace}

\newcommand{\hyltlfgone}{{HyperLTL$^1(\mathbf{F},\mathbf{G})$}\xspace}
\newcommand{\hyltlfgxone}{{HyperLTL}$^1(\mathbf{F},\mathbf{G},\mathbf{X}^*)$\xspace}
\newcommand{\qptl}{QPTL\xspace}
\newcommand{\qbf}{QBF\xspace}

\newcommand{\suffix}[2]{#1[#2,\infty)}
\newcommand{\var}{\mathcal{V}}

\newcommand{\merge}{\mathrm{mrg}}

\newcommand{\M}{\mathcal{M}}

\newcommand{\np}{\textsc{NP}\xspace}
\newcommand{\pspace}{\textsc{PSpace}\xspace}

\newcommand{\nexpt}{\textsc{NExpTime}\xspace}
\newcommand{\exps}{\textsc{ExpSpace}\xspace}
\newcommand{\ntwoexpt}{\textsc{N2ExpTime}\xspace}
\newcommand{\tower}{\textsc{Tower}\xspace}

\newcommand{\TD}{\mathrm{td}}
\newcommand{\AD}{\mathrm{ad}}
\newcommand{\lord}{<}

\newcommand{\myquot}[1]{``#1''}
\newcommand{\traces}[1]{\mathrm{T}(#1)}
\def\rdots{\rotatebox[origin=l]{29}{$\scriptscriptstyle\ldots\mathstrut$}}

\newcommand{\pad}[1]{\#(#1)}

\newcommand{\window}[2]{\#_{#2}(#1)}

\newcommand{\barsig}{\overline{\Sigma}}

\newcommand{\witness}[1]{\mathrm{wtns}(#1)}

\newcommand{\twr}{\mathrm{twr}}
\newcommand{\otherone}[1]{\overline{#1}}

\newcommand{\rk}[2]{\mathrm{rk}_{#1}(#2)}

\newcommand{\spaceprop}{\mathrm{spc}}
\newcommand{\timeprop}{\mathrm{tm}}
\newcommand{\rankprop}{\mathrm{rk}}
\newcommand{\tmark}{{m_t}}
\newcommand{\smark}{{m_s}}
\newcommand{\inmark}{in}
\newcommand{\nimark}{ni}
\newcommand{\ltmark}{lt}
\newcommand{\suc}{sc}
\newcommand{\gtmark}{gt}

\newcommand{\incop}{\texttt{++}}
\newcommand{\decop}{\texttt{-\hspace{0pt}-}}
\newcommand{\testop}{\texttt{=0?}}
\usepackage{todonotes}
\usepackage{fullpage}
\usepackage{hyperref}

\title{The Keys to Decidable HyperLTL Satisfiability:\\ Small Models or Very Simple Formulas}\author{Corto Mascle\inst{1} \and Martin Zimmermann\inst{2}}

\institute{ENS Paris-Saclay, Cachan, France\\ \email{corto.mascle@ens-paris-saclay.fr} \and
University of Liverpool, Liverpool, United Kingdom\\
  \email{martin.zimmermann@liverpool.ac.uk}}

\begin{document}

\maketitle

\begin{abstract}
HyperLTL, the extension of Linear Temporal Logic by trace quantifiers, is a uniform framework for expressing information flow policies by relating multiple traces of a security-critical system.
HyperLTL has been successfully applied to express fundamental security policies like noninterference and observational determinism, but has also found applications beyond security, e.g., distributed protocols and coding theory.
However, HyperLTL satisfiability is undecidable as soon as there are existential quantifiers in the scope of a universal one.
To overcome this severe limitation to applicability, we investigate here restricted variants of the satisfiability problem to pinpoint the decidability border.

First, we restrict the space of admissible models and  show decidability when restricting the search space to models of bounded size or to finitely representable ones.
Second, we consider formulas with restricted nesting of temporal operators and show that nesting depth one yields decidability for a slightly larger class of quantifier prefixes.
We provide tight complexity bounds in almost all cases.
\end{abstract}

\section{Introduction}
\label{sec_intro}
The introduction of temporal logics for the specification of information flow policies~\cite{DBLP:conf/post/ClarksonFKMRS14} was a significant milestone in the long and successful history of applying logics in computer science~\cite{journals/bsl/HalpernHIKVV01}.
Probably the most important representative of these logics is \hyltl~\cite{DBLP:conf/post/ClarksonFKMRS14}, which extends Linear Temporal Logic (\ltl)~\cite{Pnueli/1977/TheTemporalLogicOfPrograms} by trace quantifiers.
This addition allows to express properties that relate multiple execution traces, which is typically necessary to capture the flow of information~\cite{Clarkson+Schneider/10/Hyperproperties}.
In contrast, \ltl, currently the most influential specification language for reactive systems, is only able to express properties of single traces.

\hyltl provides a uniform framework for expressing information flow policies in a formalism with intuitive syntax and semantics, and for the automated verification of these policies:
A wide range of policies from the literature~\cite{DBLP:conf/sp/GoguenM82a,DBLP:conf/sp/McCullough88,DBLP:journals/tse/McCullough90,DBLP:conf/sp/McLean94,DBLP:journals/jcs/Millen95,DBLP:conf/csfw/ZdancewicM03} with specialized verification algorithms is expressible in \hyltl, i.e., universal \hyltl verification algorithms are applicable to all of them.

As an example, consider a system with a set~$I$ of inputs, which contains a hidden input~$h \in I$, and an output~$o$. 
Now, noninterference~\cite{DBLP:conf/sp/GoguenM82a} between $h$ and $o$ requires that no information about $h$ is leaked via $o$, i.e., for all execution traces~$\pi$ and $\pi'$, if the inputs in $\pi$ and $\pi'$ only differ in $h$, then they have the same output at all times.
Formally, this is captured by the \hyltl formula
\[ \forall \pi. \forall \pi'. \left(\G \bigwedge \nolimits_{i \in I \setminus\set{h}} (i_\pi \Leftrightarrow i_{\pi'}) \right) \Rightarrow  \G (o_\pi \Leftrightarrow o_{\pi'}) .
\]

Today, there are tools for model checking \hyltl properties~\cite{hyperliveness,DBLP:conf/cav/FinkbeinerRS15}, for checking satisfiability of \hyltl properties~\cite{DBLP:conf/atva/FinkbeinerHH18,DBLP:conf/cav/FinkbeinerHS17}, for synthesizing reactive systems from \hyltl properties~\cite{DBLP:conf/cav/FinkbeinerHLST18}, and for runtime monitoring of \hyltl properties~\cite{AgrawalB16,DBLP:conf/rv/BonakdarpourF16,DBLP:conf/tacas/FinkbeinerHST18}.
Furthermore, the extraordinary expressiveness of \hyltl has been exhibited~\cite{DBLP:conf/stacs/Finkbeiner017} and connections to first and second-order predicate logics have been established~\cite{hierarchy,DBLP:conf/stacs/Finkbeiner017}.

The major drawback of \hyltl is the usual price one has to pay for great expressiveness: prohibitively high worst-case complexity.
In particular, model checking finite Kripke structures against \hyltl formulas is nonelementary~\cite{DBLP:conf/post/ClarksonFKMRS14} and satisfiability is even undecidable~\cite{DBLP:conf/concur/FinkbeinerH16}.
These results have to be contrasted with model checking and satisfiability being $\pspace$-complete for \ltl~\cite{Sistla:1985:CPL:3828.3837}, problems routinely solved in real-life applications~\cite{Kurshan2018}.

Due to the sobering state of affairs, it is imperative to find fragments of the logic with (more) tractable complexity. 
In this work, we focus on the satisfiability problem, the most fundamental decision problem for a logic.
Nevertheless, it has many applications in verification, e.g., checking the equivalence and implication of specifications can be reduced to satisfiability.
Finally, the question whether a property given by some \hyltl formula is realizable by some system is also a satisfiability problem.

A classical attempt to overcome the undecidability of the satisfiability problem is to restrict the number of quantifier alternations of the formulas under consideration.
In fact, the alternation depth is the measure underlying the nonelementary complexity of the \hyltl model checking problem~\cite{DBLP:conf/post/ClarksonFKMRS14}.
However, the situation is different for the satisfiability problem: It is undecidable even when restricted to $\forall\exists^*$ formulas, i.e., formulas starting with one universal quantifier followed by existential ones~\cite{DBLP:conf/concur/FinkbeinerH16}.
All remaining prefix classes are decidable by reductions to the \ltl satisfiability problem, e.g., the satisfiability problem is $\pspace$-complete for the alternation-free prefix classes~$\exists^*$ and $\forall^*$ and $\exps$-complete for the class~$\exists^*\forall^*$~\cite{DBLP:conf/concur/FinkbeinerH16}. 

However, there are more complexity measures beyond the alternation depth that can be restricted in order to obtain tractable satisfiability problems, both on formulas and on models.
The latter case is of particular interest, since it is known that not every satisfiable \hyltl has a \myquot{simple} model, for various formalizations of \myquot{simple}~\cite{DBLP:conf/stacs/Finkbeiner017}.
Thus, for those formulas, such a restriction could make a significant difference. 
Furthermore, from a more practical point of view, one is often interested in whether there is a, say, finite model while the existence of an intricate infinite model may not be useful. 

We study the satisfiability problem for formulas with restricted quantifier prefixes and restricted temporal depth~\cite{DBLP:journals/iandc/DemriS02}, which measures the nesting of temporal operators. 
Our main result here shows that satisfiability is even undecidable for formulas of the form $\forall^2\exists^*\phi$, where $\phi$ has temporal depth one and only uses eventually~$\F$ and always~$\G$,  i.e., it is a Boolean combination of formulas~$\F \phi'$ with propositional $\phi'$.
Thereby, we strengthen the previous undecidability result for $\forall\exists^*$ by bounding the temporal depth to one, but at the price of a second universal quantifier.
Moreover, we clarify the border between decidability and undecidability at temporal depth two:
Using only one universally quantified variable, temporal depth one, and only $\F$, $\G$, and nested applications of next~$\X$ leads to decidability. 
Finally, we show that every \hyltl formula can be transformed into an equisatisfiable $\forall^2\exists^*$ formula of temporal depth two, i.e., this fragment already captures the full complexity of the satisfiability problem.

Thus, the overall picture is still rather bleak: if one only restricts the formula then the islands of decidability are very small.
Phrased differently, even very simple formulas are extremely expressive and allow to encode computations of Turing-complete devices in their models.
However, note that such models are necessarily complex, as they need to be able to encode an unbounded amount of information.

Thus, we also consider satisfiability problems for arbitrary formulas, but with respect to restricted models which do not allow to encode such computations.
In particular, we consider three variants of increasing complexity:
Checking whether a given \hyltl formula has a model of a given cardinality~$k$ is $\exps$-complete, whether it has a model containing only ultimately periodic traces of length at most $k$ is $\ntwoexpt$-complete, and checking whether it has a model induced by a Kripke structure with $k$ states is $\tower$-complete. 
The last result is even true for a fixed Kripke structure, which therefore has implications for the complexity of the model checking problem as well.
Thus, the situation is more encouraging when checking for the existence of small models: 
satisfiability becomes decidable, even with (relatively) moderate complexity in the first two cases.

However, as argued above, all three approaches are (necessarily) incomplete:
There are satisfiable formulas that have only infinite models, satisfiable formulas that have only non-ultimately periodic models, and satisfiable formulas that have no $\omega$-regular models \cite{DBLP:conf/stacs/Finkbeiner017}, a class of models that includes all those that are induced by a finite Kripke structure.

All in all, our work shows that \hyltl satisfiability remains a challenging problem, but we have provided a complete classification of the tractable cases in terms of alternation depth, temporal depth, and representation of the model (for formulas without until).

All proofs omitted due to space restrictions can be found in the appendix.

\section{Definitions}
\label{sec_hyperltl}

Fix a finite set~$\ap$ of atomic propositions. A \emph{valuation} is a subset of $\ap$. A \emph{trace} over $\ap$ is a map $t \colon \nats \rightarrow \pow{\ap}$, denoted by $t(0)t(1)t(2) \cdots$, i.e., an infinite sequence of valuations. The set of all traces over $\ap$ is denoted by $(\pow{\ap})^\omega$. The \emph{projection} of $t$ to $\ap'$ is the trace~$(t(0) \cap \ap') (t(1) \cap \ap') (t(2) \cap \ap') \cdots $ over $\ap'$. A trace~$t$ is \emph{ultimately periodic}, if $t = x \cdot y^\omega$ for some $x,y \in (\pow{\ap})^+$, i.e., there are $s,p>0$ with $t(n) = t(n+p)$ for all $n \ge s$.

The formulas of \hyltl are given by the grammar
\begin{align*}
\phi & {} \cceq {}  \exists \pi. \phi \mid \forall \pi. \phi \mid \psi \\
\psi & {}  \cceq {}  a_\pi \mid \neg \psi \mid \psi \vee \psi \mid \X \psi \mid \psi \U \psi 
\end{align*}
where $a$ ranges over atomic propositions in $\ap$ and where $\pi$ ranges over a fixed countable set~$\var$ of \emph{trace variables}. Conjunction, implication, equivalence, and exclusive disjunction~$\oplus$, as well as the temporal operators eventually~$\F$ and always~$\G$ are derived as usual. A \emph{sentence} is a closed formula, i.e., a formula without free trace variables. The \emph{size} of a formula $\phi$, denoted by $\size{\phi}$, is its number of distinct subformulas.

The semantics of \hyltl is defined with respect to a \emph{trace assignment}, a partial mapping~$\Pi \colon \var \rightarrow (\pow{\ap})^\omega$. The assignment with empty domain is denoted by $\Pi_\emptyset$. Given a trace assignment~$\Pi$, a trace variable~$\pi$, and a trace~$t$ we denote by $\Pi[\pi \rightarrow t]$ the assignment that coincides with $\Pi$ everywhere but at $\pi$, which is mapped to $t$. We also use shorthand notation like~$[\pi_1 \rightarrow t_1,\ldots ,\pi_n \rightarrow t_n]$ and $[(\pi_i \rightarrow t_i)_{1 \le i \le n}]$ for $\Pi_\emptyset[\pi_1 \rightarrow t_1]\ldots[\pi_n \rightarrow t_n]$, if the $\pi_i$ are pairwise different. Furthermore, $\suffix{\Pi}{j}$ denotes the trace assignment mapping every $\pi$ in $\Pi$'s domain to $\Pi(\pi)(j)\Pi(\pi)(j+1)\Pi(\pi)(j+2) \cdots $.

For sets~$T$ of traces and trace assignments~$\Pi$ we define 
\begin{itemize}
	\item $(T, \Pi) \models a_\pi$, if $a \in \Pi(\pi)(0)$,
	\item $(T, \Pi) \models \neg \psi$, if $(T, \Pi) \not\models \psi$,
	\item $(T, \Pi) \models \psi_1 \vee \psi_2 $, if $(T, \Pi) \models \psi_1$ or $(T, \Pi) \models \psi_2$,
	\item $(T, \Pi) \models \X \psi$, if $(T,\suffix{\Pi}{1}) \models \psi$,
	\item $(T, \Pi) \models \psi_1 \U \psi_2$, if there is a $j \ge 0$ such that $(T,\suffix{\Pi}{j}) \models \psi_2$ and for all $0 \le j' < j$: $(T,\suffix{\Pi}{j'}) \models \psi_1$, 
	\item $(T, \Pi) \models \exists \pi. \phi$, if there is a trace~$t \in T$ such that $(T,\Pi[\pi \rightarrow t]) \models \phi$, and 
	\item $(T, \Pi) \models \forall \pi. \phi$, if for all traces~$t \in T$: $(T,\Pi[\pi \rightarrow t]) \models \phi$. 
\end{itemize}
We say that $T$ \emph{satisfies} a sentence~$\phi$ if $(T, \Pi_\emptyset) \models \phi$. In this case, we write $T \models \phi$ and say that $T$ is a \emph{model} of $\phi$. Conversely, satisfaction of quantifier-free formulas does not depend on $T$. Hence, we say that $\Pi$ \emph{satisfies} a quantifier-free~$\psi$ if $(\emptyset,\Pi) \models\psi$ and write $\Pi \models \psi$ (assuming $\Pi$ is defined on all trace variables that appear in $\psi$). 

The \emph{alternation depth} of a \hyltl formula $\phi$, denoted by $\AD (\phi)$, is defined as its number of quantifier alternations. Its \emph{temporal depth}, denoted by $\TD (\phi)$, is defined as the maximal depth of the nesting of temporal operators in the formula. Formally, $\TD$ and $\AD$ are defined as follows :
\begin{multicols}{2}
\begin{itemize}
	\item $\TD(a_\pi) = 0$
	\item $\TD(\neg \psi) = \TD(\psi)$
	\item $\TD(\psi_1 \vee \psi_2) = \max(\TD(\psi_1), \TD(\psi_2))$,
	\item $\TD(\X \psi) = 1+\TD(\psi)$,
	\item $\TD(\psi_1 \U \psi_2) = 1+\max(\TD(\psi_1), \TD(\psi_2))$, 
 	\item $\TD(\exists \pi. \phi) = \TD(\phi)$
	\item $\TD(\forall \pi. \phi) = \TD(\phi)$. 
\end{itemize}

\columnbreak

\begin{itemize}
	\item $\AD(\exists \pi_1 \ldots \exists \pi_n.\phi) = 0$ for quantifier-free~$\phi$
	\item $\AD(\forall \pi_1 \ldots \forall \pi_n.\phi) = 0$ for quantifier-free~$\phi$
	\item $\AD(\exists \pi_1\ldots\exists \pi_n. \forall \tau. \phi) = 1+\AD(\forall \tau .\phi)$
	\item $\AD(\forall \pi_1\ldots\forall \pi_n. \exists \tau. \phi) = 1+\AD(\exists \tau .\phi)$
\end{itemize}
\end{multicols}

Although \hyltl sentences are required to be in prenex normal form, they are closed under Boolean combinations, which can easily be seen by transforming such formulas into prenex normal form. Note that this transformation can be implemented such that it changes neither the temporal nor alternation depth, and can be performed in polynomial time.

The fragment \hyltlfgone contains formulas of temporal depth one using only $\F$ and $\G$ as temporal operators, and \hyltlfgxone contains formulas using only $\F$, $\G$, and $\X$ as temporal operators and of temporal depth one, however we allow iterations of the $\X$ operator. Formally, \hyltlfgxone formulas are generated by the grammar
\begin{align*}
\phi & {} \cceq {}  \exists \pi. \phi \mid \forall \pi. \phi \mid \psi \\
\psi & {}  \cceq {}  \neg \psi  \mid \psi \vee \psi \mid \psi \wedge \psi \mid \X^n \psi' \mid \F \psi' \mid \G \psi' \mid \psi'  \\
\psi' & {} \cceq {} a_\pi \mid \neg \psi' \mid  \psi' \lor \psi' \mid \psi' \wedge \psi'
\end{align*}
where $n$ ranges over the natural numbers. The grammar for \hyltlfgone is obtained by removing $\X^n\psi'$ from the grammar above.

Also, we use standard notation for classes of formulas with restricted quantifier prefixes, e.g., $\forall^2\exists^*$ denotes the set of \hyltl formulas in prenex normal form with two universal quantifiers followed by an arbitrary number of existential quantifiers, but no other quantifiers.

Finally, we encounter various complexity classes, classical ones from \np to \ntwoexpt, as well as \tower (see, e.g., \cite{DBLP:journals/toct/Schmitz16}). 
Intuitively, \tower is the set of problems that can be solved by a Turing machine that, on an input of size $n$, stops in time 
$2^{2^{\rdots^2}}$, with the height of the tower of exponents bounded by $b(n)$, where $b$ is a fixed elementary function. The reductions presented in this work are polynomial time reductions unless otherwise stated.

\section{Satisfiability for Restricted Classes of Models}
\label{sec_models}
As the satisfiability problem \myquot{Given a \hyltl sentence $\varphi$, does $\varphi$ have a nonempty model?} is undecidable, even when restricted to finite models~\cite{DBLP:conf/concur/FinkbeinerH16}. Hence, one has to consider simpler problems to regain decidability. In this section, we simplify the problem by checking only for the existence of \emph{simple} models, for the following three formalizations of simplicity, where the bound~$k$ is always part of the input:
\begin{itemize}
    
    \item Models of cardinality at most $k$ (Theorem~\ref{ktraces}).
    
    \item Models containing only ultimately periodic traces~$x y^\omega$ with $\size{x}+\size{y} \leq k$ (Theorem~\ref{kperiod}).

    \item Models induced by a finite-state system with at most $k$ states (Theorem~\ref{kstates}).

\end{itemize}
In every case, we allow arbitrary \hyltl formulas as input and encode $k$ in binary.

With the following result, we determine the complexity of checking satisfiability with respect to models of bounded size. The algorithm uses a technique introduced by Finkbeiner and Hahn \cite[Theorem 3]{DBLP:conf/concur/FinkbeinerH16} that allows us to replace existential and universal quantification by disjunctions and conjunctions, if the model is finite. Similarly, the lower bound also follows from Finkbeiner and Hahn.

\begin{theorem}\label{ktraces}
The following problem is \exps-complete: 
Given a \hyltl sentence~$\phi$ and $k \in \nats$ (in binary), does $\phi$ have a model with at most $k$ traces?
\end{theorem}

\begin{proof}
For the \exps upper bound, one can check, given $\phi$ and $k$, satisfiability of the sentence $\exists \pi_1\ldots\exists \pi_k.\overline{\phi}$ where $\overline{\phi}$ is defined inductively as follows:
\begin{itemize}

    \item $\overline{\phi} = \phi$ if $\phi$ is quantifier-free.

    \item $\overline{\forall \pi. \phi} = \bigwedge_{i=1}^k \overline{\phi}[\pi \leftarrow \pi_i]$.
    
    \item $\overline{\exists \pi. \phi} = \bigvee_{i=1}^k \overline{\phi}[\pi \leftarrow \pi_i]$.

\end{itemize}
Here, $\overline{\phi}[\pi \leftarrow \pi_i]$ is obtained from $\overline{\phi}$ by replacing every occurrence of $\pi$ by $\pi_i$. This sentence states the existence of at most $k$ traces satisfying $\phi$ by replacing every quantifier by an explicit conjunction or disjunction over the possible assignments.

The resulting sentence is of size at most $\size{\phi}k^{\size{\phi}} + k$, which is exponential in the size of the input and its satisfiability can be checked in polynomial space in the size of the resulting formula~\cite{DBLP:conf/concur/FinkbeinerH16}. As a result, the problem is in \exps as well.

Finkbeiner and Hahn showed that satisfiability is \exps-complete for sentences of the form $\exists^*\forall^*$~\cite{DBLP:conf/concur/FinkbeinerH16}.
This implies \exps-hardness of our problem, as if such a sentence, say with $k$ existential quantifiers, is satisfiable then it has a model with at most $k$ traces. 
\qed
\end{proof}

As the algorithm proceeds by a reduction to the satisfiability problem for $\exists^*$ formulas, which in turn is reduced to \ltl satisfiability, one can show that a \hyltl sentence~$\varphi$ has a model with $k$ traces if and only if it has a model with $k$ ultimately periodic traces.

Next, we consider another variant of the satisfiability problem, where we directly restrict the space of possible models to ultimately periodic ones of the form~$xy^\omega$ with $\size{x} + \size{y} \le k$. As we encode $k$ in binary, the length of those traces is exponential in the input and the cardinality of the model is bounded doubly-exponentially. 
This explains the increase in complexity in the following theorem in comparison to Theorem~\ref{ktraces}.

\begin{theorem}\label{kperiod}
The following problem is \ntwoexpt-complete: 
Given a \hyltl sentence~$\phi$ and $k \in \nats$ (in binary), does $\phi$ have a model whose elements are of the form $xy^\omega$ with $\size{x}+\size{y} \leq k$?
\end{theorem}

\begin{proof}
For the upper bound, given a \hyltl sentence $\phi$ and $k \in \nats$ in binary we start by guessing a model $T \subseteq \set{xy^\omega \mid \size{x}+\size{y}\leq k}$. Let $n = \lceil \log_2(k) \rceil + \size{\phi}$ be the size of the input. We have $\size{T}\leq (2^{\size{\ap}} +1)^{k+1} \leq (2^n+1)^{2^n + 1}$ (as we can assume all atomic propositions appear in the formula, $\size{\ap} \leq \size{\phi} \leq n$), i.e., $\size{T}$ is doubly-exponential in $n$.

Then, we apply to $\phi$ a similar transformation as the one used in the proof of Theorem \ref{ktraces}, i.e., we create a variable~$\pi_{u,v}$ for each trace~$uv^\omega \in T$, and we replace in the sentence every universal (existential) quantifier by a conjunction (disjunction) over every possible trace assignment of the quantified variable over $T$. Formally, we define $\overline{\phi}$ as follows:
\begin{itemize}
    \item $\overline{\phi} = \phi$ if $\phi$ is quantifier-free
    \item $\overline{\forall \pi. \phi} = \bigwedge_{uv^\omega \in T} \overline{\phi}[\pi \leftarrow \pi_{u,v}]$
    \item $\overline{\exists \pi. \phi} = \bigvee_{uv^\omega \in T} \overline{\phi}[\pi \leftarrow \pi_{u,v}]$
\end{itemize}

The size of the formula is multiplied by $\size{T}$ at each new quantifier. In the end, the size of $\overline{\phi}$ is at most $\size{T}^{\size{\phi}} \size{\phi} \leq n((2^n+1)^{2^n + 1})^n = 2^{2^{O(n)}}$. Define $\Pi_T = [ (\pi_{u,v} \rightarrow uv^\omega)_{uv^\omega \in T} ]$.

Observe that as $T \subseteq \set{xy^\omega \mid \size{x}+\size{y}\leq k }$, for all $j>k$, $w(j+k!) = w(j)$ for every $w \in T$. Thus, every quantifier-free formula evaluated over $\Pi_T$ is satisfied at some index if and only if it is satisfied at some index less than $k! +k$. Therefore, one can evaluate $\overline{\phi}$ over $\Pi_T$ by recursively computing for each subformula the set of indices $0 \leq j < k+k!$ such that this subformula is satisfied by $\Pi_T[j,\infty)$. This procedure is polynomial in the size of $\overline{\phi}$, $T$ and $k!$, all three being doubly-exponential in $n$, i.e., it is doubly-exponential in the size of the input.

To prove the lower bound, we reduce from the following bounded variant of Post's correspondence problem (PCP): Given two lists~$u_1, \ldots, u_n $ and $u_1', \ldots, u_n' $ of words, does there exist a word~$s \in \set{1,\ldots, n}^+$ with $\size{s} \le \pow{\pow{n}}$ and $h(s) = h'(s)$? Here $h$ is the homomorphism induced by~$m \mapsto u_m$ for every $m$ and $h'$ is defined analogously. 
We refer to such a word~$s$ as a solution.
This variant is \ntwoexpt-hard, which can be shown by adapting the proof of undecidability of PCP as presented by Hopcroft and Ullman~\cite{DBLP:books/lib/HopcroftU69}. The same construction allows to reduce from the problem of deciding, given a nondeterministic Turing machine $\M$ and a word $w$, if $\M$ halts on $w$ in at most $\pow{\pow{\size{w}}}$ steps. 

Now, given an instance~$u_1, \ldots,u_n, u_1',\ldots,u_n'$ of this problem over an alphabet~$\Sigma$, we construct a sentence~$\phi$ and a $k \in \nats$ in polynomial time in $\sum_m  \size{u_m}+ \size{u_m'}$ such that $\phi$ has a model $T \subseteq \set{xy^\omega \mid \size{x}+\size{y}\leq k}$ if and only if the instance has a solution. 

Let $\ell$ be the maximal length of a word of the instance. For all $1\leq m \leq n$, let $\pad{u_m} = u_m\#^{\ell-\size{u_m}}$ and $\pad{u_m'}= u_m'\#^{l-\size{u_m'}}$, i.e., $\size{\pad{u_m}} = \size{\pad{u_m'}} = \ell$ for all $m$.
For every $a \in \Sigma$, let $\overline{a}$ be a fresh letter and let $\barsig = \set{\overline{a} \mid a \in \Sigma}$. For all $1 \leq m \leq n$ and $0\leq j < \size{u_m}$ let 
\[\window{u_m}{j} = \pad{u_m}(0)\cdots \pad{u_m}(j-1)\overline{\pad{u_m}(j)}\pad{u_m}(j+1)\cdots \pad{u_m}(\ell -1)\]
be the word $\pad{u_m}$ where the letter at position~$j$ is replaced by its associated letter in $\barsig$. Analogously, we define $\window{u_m'}{j}$. 

Further, for all $r\in \nats$, let $b_r \in \set{0,1}^*$ be the binary representation of $r$, with the least significant bit at the beginning and without trailing zeros. Finally, for all $b_r, b_{r'}$ let $[b_r,b_{r'}]$ be the unique word over $\set{0,1}^2$ such that $[b_r,b_{r'}](i) = (b_r(i),b_{r'}(i))$ for all $i$. If $b_r$ and $b_{r'}$ are not of the same length, then we pad the shorter one with $0$'s at the end.

In the following, we work with the set~$\ap=\Sigma \cup \barsig \cup \set{\#,\$,0,1} \cup \set{0,1}^2$ of atomic propositions. In the construction, we ensure that on every trace of the model exactly one proposition is satisfied at any moment in time. Thus, we treat traces as words over $\ap$ without making a distinction between a proposition~$a$ and the singleton $\set{a}$.

We define two types of traces: 
\begin{itemize}
    
    \item Traces of type one are of the form $\pad{u_m} \pad{u'_m} b_r\$^\omega$ for some $1\leq m \leq n$ and some $r \in \nats$, which is called the \emph{rank} of the trace.
    
    \item Traces of type two are of the form $(\window{u_{m}}{j})(\window{u_{m'}'}{j'})[b_r,b_{r'}]\$^\omega$ with $1\leq m,m' \leq n$, $0 \leq j < \size{u_{m}}$, $0 \leq j' < \size{u_{m'}'}$ and $r,r' \in \nats$. Note that we allow $m \neq m'$.
\end{itemize}

A trace~$t_1 = \pad{u_{m_1}} \pad{u'_{m_1}} b_{r_1}\$^\omega$ of type one and a trace~$t_2 = (\window{u_{m_2}}{j})(\window{u_{m'_2}'}{j'})[b_{r_2},b_{r'_2}]\$^\omega$ of type two are \emph{$u$-compatible} if $u_{m_1} = u_{m_2}$ and $r_1 = r_2$ and \emph{$u'$-compatible} if $u'_{m_1} = u'_{m'_2}$ and $r_1 = r'_2$.

Fix a trace~$t_0 = (\window{u_{m}}{j})(\window{u_{m'}'}{j'})[b_r,b_{r'}]\$^\omega$ of type two. 
It is \emph{final} if $j = \size{u_{m}}-1$, $j' = \size{u_{m'}'}-1$ and $b_r = b_{r'}$.
If it is not final, then a trace~$t_1$ of type two is a \emph{successor} of $t_0$, if one of the following holds:
\begin{itemize}
    
    \item $j < \size{u_{m}}-1$ and $j' < \size{u_{m'}'}-1$ and $t_1 = (\window{u_{m}}{j+1})(\window{u_{m'}'}{j'+1})[b_{r},b_{r'}]\$^\omega$.
    
    \item $j = \size{u_{m}}-1$ and $j' < \size{u_{m'}'}-1$ and $t_1 = (\window{u_{m^*}}{0})(\window{u_{m'}'}{j'+1})[b_{r+1},b_{r'}]\$^\omega$ for some $1\leq m^* \leq n$.
    
    \item $j < \size{u_{m}}-1$ and $j' = \size{u_{m'}'}-1$ and $t_1 = (\window{u_{m}}{j+1})(\window{u_{m^*}'}{0})[b_{r},b_{r'+1}]\$^\omega$ for some $1\leq m^* \leq n$.
    
    \item $j = \size{u_{m}}-1$ and $j' = \size{u_{m'}'}-1$ (which implies $r \neq r'$, as $t_0$ is not final) and $t_1 = (\window{u_{m^*}}{0})(\window{u_{{m^*}'}'}{0})[b_{r+1},b_{r'+1}]\$^\omega$ for some $1\leq m^*,{m^*}' \leq n$.

\end{itemize}

Now, we construct $\phi$. It is the conjunction of sentences expressing the following statements, which can be brought into prenex normal form. In order to improve readability, the construction of the sentences expressing each requirement is left to the reader.

\begin{enumerate}
    
    \item\label{kperiod_singleprop} Exactly one atomic proposition is satisfied at each position of each trace.
    
    \item\label{kperiod_typeness} Every trace is either of type one or of type two.
    
    \item\label{kperiod_ranks} No two distinct type one traces have the same rank.
    
    \item\label{kperiod_predecessor} For every type one trace $t$ of positive rank $r$ there exists a type one trace $t'$ of rank $r-1$.
    
    \item\label{kperiod_initial} There exists a type two trace of the form $(\window{u_{m}}{0})(\window{u_{m}'}{0})[b_0,b_0]\$^\omega$ for some $m$.

    \item\label{kperiod_successor} For every type two trace $t_0$ there exists a type two trace $t_1$ such that either $t_0$ is final or $t_1$ is a successor of $t_0$.
    
    \item\label{kperiod_consistent} For every type two trace~$t_2$ there exist two type one traces~$t_1$ and $t_1'$ such that $t_2$ and $t_1$ are $u$-compatible and $t_2$ and $t_1'$ are $u'$-compatible.
    
    \item\label{kperiod_sameletters} In every trace~$t = (\window{u_{m}}{j})(\window{u_{m'}'}{j'})[b_{r},b_{r'}]\$^\omega$  of type two, the letters from $\barsig$ in $u_{m}$ and $u_{m'}'$ are equal.

\end{enumerate}

As an example, consider the following instance of PCP
\begin{center}
\begin{tabular}{ll}
$u_1 = b$ &\mbox{}\qquad\qquad $u_2 = aa$\\
$u_1' = aba $&\mbox{}\qquad\qquad $u_2' = a$
\end{tabular}
\end{center}
with solution~$212$. We encode the solution by the set~$T_1 \cup T_2$ depicted below (where we drop the suffixes~$\$^\omega$ for readability).
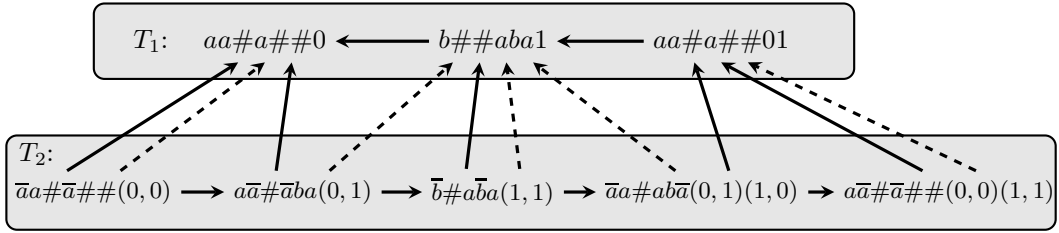
\begin{center}
    \begin{tikzpicture}
    
        \draw[rounded corners, fill=black!10, thick] (-6,.-.5) rectangle (4,.5);

        \draw[rounded corners, fill=black!10, thick] (-7.15,.-2.5) rectangle (6.65,-1.25);

\node at (-5.25, 0) {$T_1$:};
\node at (-6.75, -1.5) {$T_2$:};

    \node (10) at (-3.75,0) {$ aa\#a\#\#0 $};
    \node (11) at (-.75,0) {$b\#\#aba1 $};
    \node (12) at (2.25,0) {$aa\# a\#\#01 $};
    
    \node (20) at (-6,-2) {\footnotesize$ \overline{a}a\#\overline{a}\#\#(0,0)  $};
    \node (21) at (-3.25,-2) {\footnotesize$ a\overline{a}\#\overline{a}ba(0,1)  $};
    \node (22) at (-.75,-2) {\footnotesize$ \overline{b}\#a\overline{b}a(1,1)  $ };
    \node (23) at (2,-2) {\footnotesize$ \overline{a}a\#ab\overline{a}(0,1)(1,0)  $ };
    \node (24) at (5.25,-2) {\footnotesize$ a\overline{a}\#\overline{a}\#\#(0,0)(1,1)  $ };
    
    \path[-stealth,very thick]
    (11.west) edge[] (10.east)
    (12.west) edge[] (11.east)
    (20.east) edge[] (21.west)
    (21.east) edge[] (22.west)
    (22.east) edge[] (23.west)
    (23.east) edge[] (24.west)
    ([xshift=-10pt]20.north) edge[] ([xshift=-10pt]10.south)
    ([xshift=-10pt]21.north) edge[] ([xshift=10pt]10.south)
    ([xshift=-10pt]22.north) edge[] ([xshift=-5pt]11.south)
    ([xshift=10pt]23.north) edge[] ([xshift=-10pt]12.south)
    ([xshift=-10pt]24.north) edge[] (12.south)
    ([xshift=10pt]20.north) edge[dashed] (10.south)
    ([xshift=10pt]21.north) edge[dashed] ([xshift=-15pt]11.south)
    ([xshift=10pt]22.north) edge[dashed] ([xshift=5pt]11.south)
    ([xshift=-10pt]23.north) edge[dashed] ([xshift=15pt]11.south)
    ([xshift=10pt]24.north) edge[dashed] ([xshift=10pt]12.south)
;

    \end{tikzpicture}
\end{center}
Note that $T_1$ contains type one traces while $T_2$ contains type two traces.
We claim that $T_1 \cup T_2$ satisfies all eight requirement listed above. Requirements~\ref{kperiod_singleprop}, \ref{kperiod_typeness}, \ref{kperiod_ranks}, and \ref{kperiod_sameletters} are straightforward to verify. The relation between type one traces specified by Requirement~\ref{kperiod_predecessor} is given by the arrows between the type one traces in $T_1$. The initial trace as in Requirement~\ref{kperiod_initial} is the left most type two trace in $T_2$ and the successor relation on type two traces as in Requirement~\ref{kperiod_successor} is indicated by the arrows between the type two traces. The compatible type one traces as in Requirement~\ref{kperiod_consistent} are given by the arrows from $T_2$ to $T_1$ with the solid arrows denoting the $u$-compatible type one traces and the dashed arrows denoting the $u'$-compatible ones.

After this example, we prove that $\varphi$  has the desired properties in general.
To this end, fix $k = \pow{n}+2m+1$, which can be encoded in binary with polynomially many bits in $n+m$.

First, assume there is a solution~$s \in \set{1, \ldots, n}^+$, i.e., we have $h(s) = h'(s)$, where $h$ and $h'$ are the homomorphisms induced by mapping $m$ to $u_m$ and $u_m'$, respectively. 
We define 
\[
T_1 = \set{ \pad{u_{s(r)}} \pad{u_{s(r)}'} b_r \$^\omega \mid 0 \leq r < \size{s}}
\] 
and 
\[
T_2 = \set{ t_{i} \mid 0 \le i < \size{h(s)} = \size{h'(s)} }
\]
where 
\[
t_{i} = (\window{u_{s(r)}}{j})(\window{u_{s(r')}'}{j'})[b_r,b_{r'}]\$^\omega
\]
where $r, j$ are the unique indices such that $r < \size{s}$, $j < \size{u_{s(r)}}$, and
\[
 \sum_{r_0 < r} \size{h(s(r_0))} + j = i
\]
and, similarly, $r', j'$ are the unique indices such that $r' < \size{s}$, $j' < \size{u'_{s(r')}}$, and
\[
 \sum_{r_0 < r'} \size{h'(s(r_0))} +j' = i .
\]
Intuitively, for every $i$, we pick the pair~$(u_{s(r)}, u'_{s(r')})$ of words of the instance the $i$-th letter of $h(s)$ is in, and mark its position correctly in those words.

Traces of $T_1$ are of type one and traces of $T_2$ are of type two. It is then easy to check that Requirements \ref{kperiod_singleprop} to \ref{kperiod_sameletters} are satisfied by $T = T_1 \cup T_2$. Furthermore, as $\size{s} \leq \pow{\pow{n}}$, the size of the $b_r$ is at most $\pow{n}$. Thus, all words of $T$ are of the form $x\$^\omega$ with $\size{x} \leq \pow{n}+2m = k-1$, as required.

Conversely, assume $\phi$ is satisfied by a set~$T \subseteq \set{xy^\omega \mid \size{x}+\size{y} \leq k}$ of traces. Then, by Requirement~\ref{kperiod_typeness}, we have $T = T_1 \cup T_2$ where $T_1$ is a set of type one traces and $T_2$ is a set of type two traces. Let $r^*$ be the maximal rank of a trace in $T_1$ (which is well-defined as $T_1$ is finite). Then, by Requirements~\ref{kperiod_ranks} and \ref{kperiod_predecessor}, there exists exactly one trace $t^1_r \in T_1$ such that $t^1_r$ is of rank $r$, for all $0 \leq r \leq r^*$. Let $s^* \in \set{1,...,n}^{r^*+1}$ be such that for all $r \leq r^*$, $t^1_r = \pad{u_{s^*(r)}}\pad{u'_{s^*(r)}}b_r\$^\omega$. Note that as $t^1_{r^*} = \pad{u_{s({r^*})}}\pad{u'_{s({r^*})}}b_{r^*}\$^\omega \in T$, we must have $2m + \size{b_{r^*}}+1 \leq k$. Thus, $\size{b_{r^*}} \leq \pow{n}$, and $r^* < \pow{\pow{n}}$. 

By Requirements~\ref{kperiod_initial} and \ref{kperiod_successor}, there exists a sequence of traces $t^2_0,t^2_1,\ldots,t^2_{p} \in T_2$ such that $t^2_{i +1}$ is a successor of $t^2_{i}$ for all $0 \le i \le p$, such that $t^2_0 = (\window{u_{m}}{0})(\window{u'_{m}}{0})[b_0,b_0]\$^\omega$ for some $m$, and such that $t^2_{p}$ is final.

An induction over $i$, using Requirements \ref{kperiod_successor} and \ref{kperiod_consistent}, shows that the first letter with a bar in $t^2_{i}$ is the $i$-th letter of $u_{s^*(0)}\cdots u_{s^*(r^*)}$ and the second one is the $i$-th letter of $u'_{s^*(0)}\cdots u'_{s^*(r^*)}$. Requirement \ref{kperiod_sameletters} ensures that those two letters are always equal. As $t^2_{p}$ is final, there exists a prefix $s$ of $s^*$ such that $\size{u_{s(0)}\cdots u_{s(\size{s}-1)}} = \size{u'_{s(0)}\cdots u'_{s(\size{s}-1)}} = p+1$ and $u_{s(0)}\cdots u_{s(\size{s}-1)}(i) = u'_{s(0)}\cdots u'_{s(\size{s}-1)}(i)$ for all $i$, i.e., the two words are equal. As $\size{s} \leq \size{s^*} \leq \pow{\pow{n}}$, we therefore obtain a solution to the PCP instance. 
\qed
\end{proof}

As expected, the complexity of the satisfiability problem increases the more traces one has at hand to encode computations. 
In Theorem~\ref{ktraces}, we have exponentially many; in Theorem~\ref{kperiod}, we have doubly-exponentially many.
In our last theorem, we consider infinite sets of traces that are finitely representable by finite-state systems. 
Here, satisfiability becomes intractable, yet still decidable, even when restricted to formulas of temporal depth one.

Formally, a \emph{Kripke structure}~$\K = (Q, \delta, q_0, \lambda)$ consists of a finite set~$Q$ of states, an initial state~$q_0 \in Q$, a transition function~$\delta \colon Q \rightarrow 2^Q \setminus \set{\emptyset}$ for all $q$, and a labelling function~$\lambda \colon Q \rightarrow \pow{\ap}$. A \emph{run} of $\K$ is an infinite sequence~$q_0q_1q_2\cdots$ of states starting with  $q_0$ and such that $q_{j+1} \in \delta(q_j)$ for all $j \in \nats$. A trace of $\K$ is the sequence of labels~$\lambda(q_0) \lambda(q_1) \lambda(q_2)\cdots$ associated to a run $q_0q_1q_2\cdots$ of $\K$. The set of traces of $\K$ is denoted by~$\traces{\K}$.

\begin{theorem}\label{kstates}
The following problem is \tower-complete: 
Given a \hyltl sentence $\phi$ and $k \in \nats$ (in binary), does $\varphi$ have a model~$T(\K)$ for some Kripke structure~$\K$ with at most $k$ states?
\end{theorem}

\begin{proof}
 Clarkson et al.\ presented a model-checking algorithm for \hyctlstar (and thus for \hyltl, which is a fragment of \hyctlstar), and showed that its complexity is a tower of exponentials whose height is the alternation depth of the input sentence~\cite{DBLP:conf/post/ClarksonFKMRS14}. Thus, one can enumerate all Kripke structures with at most $k$ states (up to isomorphism) and model-check them one by one in $\tower$. This yields the desired upper bound, as there are \myquot{only} exponentially many (in $k$) Kripke structures with $k$ states.
 
 The lower bound is obtained by a reduction from the universality problem for star-free regular expressions with complementation. The equivalence problem for those expressions is \tower-complete (under elementary reductions, which is standard for $\tower$-complete problems), even for two-letter alphabets \cite{DBLP:journals/toct/Schmitz16,DBLP:conf/stoc/StockmeyerM73}. As those expressions are closed by complementation and union, the universality problem is \tower-complete as well. 
 
 Star-free expressions with complementation over~$\set{a,b}$ are generated by the grammar
\[e \cceq a \mid  b  \mid \epsilon \mid \emptyset \mid e+e \mid ee \mid \neg e \]
and have the obvious semantics inducing a language over $\set{a,b}^*$, denoted by $e$ as well.

Let~$e$ be such an expression. We construct a \hyltl sentence~$\phi_e$ and a Kripke structure~$\K$ such that $T(\K)$ is a model of $\phi_e$ if and only if $e$ is universal. $\K$ does not depend on $e$ and is shown in Figure \ref{fig:KripkeTOWER}. 
As all sets of variables in $\K$ are singletons, we indifferently use the notation~$a$ for the letter~$a$ and the singleton~$\set{a}$. The set of traces induced by this Kripke structure is 
\[T(\K)=l^\omega + l^*(a+b)^\omega + l^*(a+b)^*r^\omega + l^*\#r^\omega.\] 

\begin{figure}[t]
    \centering
\begin{tikzpicture}[->,>=stealth',shorten >=1pt,auto,node distance=2cm,
                    thick]

  \node[state]         (A) at (0,1.05)   {$a$};
  \node[state]         (B) at (0,-1.05){$b$};
  \node[state]         (L) at (-3,0)  {$l$};
  \node[state]         (R) at (3,0)   {$r$};
  \node[state]         (S) at (0,2.1)   {$\#$};

  \path (A) edge [loop below]   node {} (A)
        (A.south east) edge [bend left]    node {} (B.north east)
        (A)    edge                node {} (R)
        (B) edge [loop above]   node {} (B)
        (B) edge [bend left]    node {} (A.south west)
        (B) edge []   node {} (R)
        (L) edge [loop left]    node {} (L)

            edge                node {} (A)
            edge []   node {} (B)
        (L.south) edge [bend right=45] node {} (R.south)
        (L.north)    edge [bend left]    node {} (S.west)
        (R) edge [loop right]   node {} (R)
        (S.east) edge [bend left]    node {} (R.north);
\end{tikzpicture}
    \caption{The Kripke structure $\K$ (all states are initial).}
    \label{fig:KripkeTOWER}
\end{figure}
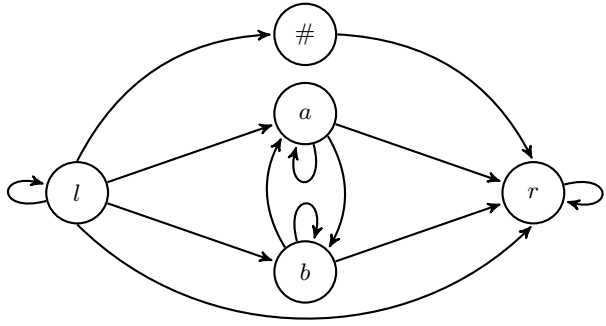

Given an expression~$e$ and a trace variable~$\pi$, we inductively define a formula $\psi_{e,\pi}$ which expresses that when $\pi$ is mapped by a trace assignment $\Pi$ to a trace of $\K$ of the form $l^nwr^\omega$ with $w\in \set{a,b}^*$, then $w \in e$ if and only if $(T(\K), \Pi) \models \psi_{e,\pi}$.

\begin{itemize}
    \item $\psi_{\emptyset,\pi} = a_\pi \land \neg a_\pi $:
    No trace assignment satisfies $\psi_{\emptyset,\pi}$, just as the language of $\emptyset$ does not contain any word.
    
    \item $\psi_{\epsilon,\pi} = \G(l_\pi \lor r_\pi)$:
    $(T(\K),\Pi)$ with $\Pi(\pi) = l^nwr^\omega$ satisfies $\psi_{\epsilon,\pi}$ if and only if $w = \epsilon$.
    
    \item $\psi_{a,\pi} = \exists \tau. (\F \#_\tau) \land \F(a_\pi) \land \G(l_\tau \Leftrightarrow l_\pi \land r_\tau \Leftrightarrow r_\pi)$ :
    The traces of $\K$ with an occurrence of $\#$ are the traces of the form $l^*\#r^\omega$.  Thus, $(T(\K),\Pi)$ with $\Pi(\pi) = l^nwr^\omega$ satisfies $\psi_{a,\pi}$ if and only if $l^nwr^\omega$ is a copy of such a trace with $\#$ replaced by $a$, i.e., if and only if $w = a$.
    
    \item $\psi_{b,\pi} = \exists \tau. (\F \#_\tau) \land \F(b_\pi) \land \G(l_\tau \Leftrightarrow l_\pi \land r_\tau \Leftrightarrow r_\pi)$: Similarly to $\psi_{a,\pi}$.
    
    \item $\psi_{e_1+e_2,\pi} = \psi_{e_1,\pi} \lor \psi_{e_2,\pi} $.
    
    \item $\psi_{e_1e_2,\pi} = \exists \pi_1. \exists \pi_2. \psi \wedge \psi'$ with
    \[\psi = \F r_{\pi_1} \land \F r_{\pi_2} \land \G (\neg \#_{\pi_1} \land \neg \#_{\pi_2}) \land \psi_{e_1,\pi_1} \land \psi_{e_2,\pi_2}
    \]
    expressing that $\pi_1$ and $\pi_2$ are of the form $l^{n_1}w_1r^\omega$ and $l^{n_2}w_2r^\omega$ with $w_1 \in e_1$ and $w_2 \in e_2$, and with
    \[\psi' = \G(l_{\pi_2} \Leftrightarrow \neg r_{\pi_1})\land \G(a_\pi \Leftrightarrow (a_{\pi_1} \lor a_{\pi_2}) \land b_\pi \Leftrightarrow (b_{\pi_1} \lor b_{\pi_2}))
    \]
    expressing that $n_2=n_1+\size{w_1}$ and that $w = w_1w_2$, where $\Pi(\pi) = l^nwr^\omega$.
 Thus, $(T(\K),\Pi)$ satisfies $\psi_{e_1e_2,\pi}$ if and only if there exist $w_1 \in e_1, w_2 \in e_2$ such that $w = w_1w_2$.
    
    \item $\psi_{\neg e,\pi} = \neg \psi_{e,\pi}$.
\end{itemize}

Although this inductive definition does not necessarily give a formula in prenex normal form, one can easily check that no quantifier is in the scope of a temporal operator, thus the resulting formula can be turned into a \hyltl formula. 

To conclude, consider the sentence~$\phi_e = \forall \pi. \G \neg r_\pi \lor \F\#_\pi \lor \psi_{e,\pi}$, which can again be brought into prenex normal form. Further, note that no temporal operator is in the scope of another one, thus $\phi_e$ has temporal depth one. The set~$T(\K)$ is a model of $\phi_e$ if and only if all its traces are in $\set{a,b,l}^\omega$, in $l^*\#r^\omega$, or of the form~$l^*wr^\omega$ with $w \in e$. 
This is the case if and only if all words $w \in \set{a,b}^*$ are in the language of $e$, i.e., if and only if $e$ is universal.
\qed
\end{proof}

As the Kripke structure~$\K$ in the lower bound proof above is fixed, we also obtain a novel hardness result for model-checking. 

\begin{corollary}\label{MCtower}
\hyltl model-checking a fixed Kripke structure with five states is \tower-complete, even for sentences of temporal depth one.
\end{corollary}

Note that one could already infer the \tower-completeness of the model-checking problem by carefully examining the proof of Theorem 5 of  \cite{DBLP:conf/post/ClarksonFKMRS14} concerning \hyctlstar model-checking. The reduction from the satisfiability problem for \qptl presented there also works for \hyltl, albeit with temporal depth larger than one. Interestingly, both reductions use a fixed Kripke structure, meaning in particular that the encoding of $k$ has no impact on the asymptotic complexity.


\section{Satisfiability for Restricted Classes of Formulas}
\label{sec_formulas}

After studying the \hyltl satisfiability problem for classes of restricted models, but arbitrary formulas, we now consider restrictions on formulas, but arbitrary models. 
Recall that Finkbeiner and Hahn presented a complete picture in terms of quantifier prefixes:
Satisfiability is $\pspace$-complete for the alternation-free fragments $\exists^*$ and $\forall^*$ as well as $\exps$-complete for $\exists^*\forall^*$.
In all other cases, the problem is undecidable, i.e., as soon as there is a single universal quantifier in front of existential ones. 

In a sense, the decidable fragments are variants of \ltl: Both alternation-free fragments can easily be reduced to \ltl satisfiability while the $\exists^*\forall^*$ one is easily reducible to the $\exists^*$ fragment, with an exponential blowup. 
Thus, the decidable fragments barely exceed the realm of \ltl.

In this section, we consider another dimension to measure the complexity of formulas, temporal depth, i.e., we restrict the nesting of temporal operators. 
The hope is that in this setting, we can obtain decidability for larger quantifier prefix classes.
However, a slight adaptation of Finkbeiner and Hahn's undecidability result for $\forall\exists^*$, along with an application of Lemma \ref{smallTD} proven below, already shows undecidability for $\forall\exists^*$ formulas of temporal depth two and without untils. 

Thus, we have to restrict our search to fragments of temporal depth one, which contain most of the information flow policies expressible in \hyltl~\cite{DBLP:conf/post/ClarksonFKMRS14}.
And indeed, we prove satisfiability decidable for $\exists^* \forall \exists^*$ \hyltlfgone formulas.
Thus, if the temporal depth is one and untils are excluded, then one can allow a universal quantifier in front of existential ones without losing decidability.
This fragment includes, for example, the noninference property~\cite{DBLP:conf/sp/McLean94}.

However, even allowing the smallest possible extension, i.e., adding a second universal quantifier, leads again to undecidability: \hyltl satisfiability is undecidable for $\forall^2\exists^*$ formulas of temporal depth one using only $\F$ as temporal operator. 
Thus, satisfiability remains hard, even when severely restricting the temporal depth of formulas. 
Our results for temporal depth one are summarized in Table~\ref{table_formulas}. 

\begin{table}
\caption{Complexity of \hyltl satisfiability in terms of quantifier prefixes and temporal depth. An asterisk~$\mbox{}^\ast$ denotes that the upper bound only holds for until-free formulas. All lower bounds in the second column already hold for temporal depth two.}
\label{table_formulas}
\centering
\begin{tabular}{lll}
 &\quad{} temporal depth one\quad{} & arbitrary temporal depth \\
 \toprule
 $ \exists^*$ / $\forall^*$ & $\np$-complete (\cite{DBLP:journals/iandc/DemriS02}+\cite{DBLP:conf/concur/FinkbeinerH16}) & $\pspace$-complete (\cite{DBLP:conf/concur/FinkbeinerH16}+\cite{Sistla:1985:CPL:3828.3837})  \\
$\exists^*\forall^*$ & $\nexpt$-complete (Thm.~\ref{nexp}) & $\exps$-complete (\cite{DBLP:conf/concur/FinkbeinerH16})\\
$\exists^*\forall \exists^*$ & in $\ntwoexpt$$^\ast$ (Thm.~\ref{sat1forall}) & undecidable (\cite{DBLP:conf/concur/FinkbeinerH16})  \\
$\forall^2 \exists^*$ & undecidable (Thm.~\ref{SatFX1}) &  undecidable \\

\end{tabular}
\end{table}

We begin this section by showing that every \hyltl formula can be transformed in polynomial time into an equisatisfiable one with quantifier prefix~$\forall^2\exists^*$ with temporal depth two.
Thus, this fragment already captures the full complexity of the satisfiability problem.
This transformation is later used in several proofs.

\begin{theorem}\label{smalldepth}
For every \hyltl sentence one can compute in polynomial time an equisatisfiable sentence of the form~$\forall^2 \exists^*$ with temporal depth at most two.
\end{theorem}

We decompose the proof into three steps, formalized by the following three lemmas.
We begin by reducing the temporal depth to at most two by adapting a construction of Demri and Schnoebelen, which associates to every \ltl formula an equisatisfiable formula with temporal depth at most two \cite{DBLP:journals/iandc/DemriS02}.

The idea is to add atomic propositions to witness the satisfaction of subformulas $\psi'$ of $\psi$. We express the existence, for every $n$-tuple of traces $(t_1,\ldots,t_n)$ of the model, of a \textit{witness trace}. For all $j \in \nats$, for all subformulas~$\psi'$ of $\psi$, the valuation $[(\pi_i \rightarrow \suffix{t_i}{j})_{1\leq i \leq n}]$ satisfies $\psi'$ if and only if the associated atomic proposition is satisfied at position~$j$ of the witness trace. 

\begin{lemma}\label{smallTD}
For every \hyltl sentence $Q_1 \pi_1 \ldots Q_n \pi_n. \psi$ with quantifier-free $\psi$, one can compute in polynomial time an equisatisfiable sentence $Q_1 \pi_1 \ldots Q_n \pi_n. \exists \pi. \psi'$ with quantifier-free $\psi'$ and temporal depth at most two.
\end{lemma}

\begin{proof}

Let $\phi = Q_1 \pi_1 \ldots Q_n \pi_n. \psi$ with quantifier-free $\psi$ and $Q_i \in \set{\exists, \forall}$ for all $i$. We denote by $\psi' \sqsubseteq \psi$ the fact that $\psi'$ is a subformula of $\psi$ and introduce a fresh atomic proposition~$m^{\psi'}$ for every $\psi' \sqsubseteq \psi$. To every trace assignment $\Pi = [(\pi_i \rightarrow t_i)_{1 \le i \le n}]$ we associate a \emph{witness trace} $\witness{t_1,\ldots,t_n}$ such that:
\begin{itemize}
    
    \item For all $a \in \ap$ and all $j \in \nats$, $a \in \witness{t_1,\ldots,t_n}(j)$ if and only if $a \in t_1(j)$, i.e., the projection of  $\witness{t_1,\ldots,t_n}$ to $\ap$ is equal to $t_1$.
    
    \item For all $\psi' \sqsubseteq \psi$ and all $j \in \nats$, $m^{\psi'} \in \witness{t_1,\ldots,t_n}(j)$ if and only if $(T,\suffix{\Pi}{j}) \models \psi'$.
\end{itemize}
Thus, $\witness{t_1,\ldots,t_n}$ is a copy of $t_1$ on which we use the $m^{\psi'}$ to mark the positions at which $\psi'$ is satisfied. 
Now, define 
\[
\phi' = Q_1 \pi_1 \ldots Q_n \pi_n. \exists \tau. m^\psi_\tau \land \bigwedge_{\psi' \sqsubseteq \psi} \overline{\psi'}.
\]
Intuitively, $\phi'$ has an additional existentially quantified variable~$\tau$ that acts as a witness and we ensure that its marking is consistent with the semantics of \hyltl using a formula $\overline{\psi'}$ for each~$\psi' \sqsubseteq \psi$, which is defined as follows:
\begin{itemize}
    \item $\overline{a_\pi} = \G(m^{a_\pi}_\tau \Leftrightarrow a_\pi)$
    \item $\overline{\psi_1 \lor \psi_2} = \G(m^{\psi_1 \lor \psi_2}_\tau \Leftrightarrow (m^{\psi_1}_\tau \lor m^{\psi_2}_\tau))$
    \item $\overline{\neg \psi_1} = \G(m^{\neg \psi_1}_\tau \Leftrightarrow \neg m^{\psi_1}_\tau)$
    \item $\overline{\X \psi_1} = \G(m^{\X \psi_1}_\tau \Leftrightarrow \X m^{\psi_1}_\tau)$
    \item $\overline{\psi_1 \U \psi_2} = \G(m^{\psi_1 \U \psi_2}_\tau \Leftrightarrow m^{\psi_1}_\tau \U m^{\psi_2}_\tau)$
\end{itemize}
 We claim that $\phi'$ has the desired properties.

Suppose~$\phi'$ is satisfied by a model~$T'$. Then, for all $t' \in T'$ let~$t$ be the projection of $t'$ to $\ap$, and let~$T$ be the set of those traces~$t$. An induction on $\psi$ allows one to prove that if \[(T', \Pi') \models \exists \tau.\, m^\psi_\tau \wedge \bigwedge_{\psi' \sqsubseteq \psi} \overline{\psi'}\] for some  trace assignment~$\Pi'$ over $T'$, then $(T',\Pi') \models \psi$. As $\psi$ only contains propositions from $\ap$, this implies $(T,\Pi)\models \psi$, where for all $\pi$, $\Pi(\pi)$ is the projection of $\Pi'(\pi)$ to $\ap$.

Then, an induction on the number of quantifiers allows to generalize this: For all $0 \leq i \leq n-1$, for all $\Pi'$, if 
\[(T',\Pi') \models Q_{n-i} \pi_{n-i} \ldots Q_n \pi_n. \exists \tau.\, m^\psi_\tau \land \bigwedge_{\psi' \sqsubseteq \psi} \overline{\psi'}
\] then $(T,\Pi) \models Q_{n-i} \pi_{n-i} \ldots Q_n \pi_n.\psi$ 
where $\Pi(\pi)$ is again the projection of $\Pi'(\pi)$ to $\ap$.
Thus, as $T'$ satisfies $\phi'$, $T$ satisfies~$\phi$. 

Now suppose $\phi$ is satisfied by a model $T$.
Let $T' = \set{\witness{t_1,\ldots,t_n} \mid t_1,\ldots,t_n \in T}$. As the projection of $\witness{t_1,\ldots,t_n}$ to $\ap$ is $t_1$, $T$ is the set of projections of traces of $T'$ to $\ap$. Furthermore, $T'$ contains a witness for all the tuples of traces of $T$, allowing us to prove that $T'$ satisfies $\phi'$ as follows.

For all $\Pi$, if $(T,\Pi) \models \psi$, then let $\Pi'$ be such that $\Pi'(\pi_i) = \witness{\Pi(\pi_i),\ldots,\Pi(\pi_i)}$ (whose projection to $\ap$ is $\Pi(\pi_i)$) for all $i$ and $\Pi'(\tau) = \witness{\Pi(\pi_1),\ldots,\Pi(\pi_n)}$. An induction on $\psi$ shows that $(T',\Pi') \models m^{\psi}_\tau \land \bigwedge_{\psi' \sqsubseteq \psi} \overline{\psi'}$. Hence, $(T',\Pi') \models \exists \tau. m^{\psi}_\tau \bigwedge_{\psi' \sqsubseteq \psi} \overline{\psi'}$.

Another induction on the number of quantifiers yields that, as $T$ satisfies $\phi$, $T'$ satisfies~$\phi'$.
Therefore, $\phi$ and $\phi'$ are equisatisfiable.
\qed
\end{proof}

Next, we turn the quantifier prefix into the form~$\forall^*\exists^*$ without increasing the temporal depth.

Here the key idea is to move existential quantifiers in the scope of universal ones after marking them with fresh atomic propositions: We can replace an $\exists \forall$ by a $\forall \exists$ if we require that the existentially quantified variable is now uniquely marked by a proposition (and therefore cannot depend on the universally quantified variable).

\begin{lemma}\label{smallAD}
For every \hyltl sentence $\phi$, one can compute in polynomial time an equisatisfiable sentence $\phi'$ of the form $\forall^*\exists^*$ with $\TD(\phi') = \max(\TD(\phi),1)$.
\end{lemma}

\begin{proof}

We call an existential quantifier in a \hyltl formula \textit{critical} if there is a universal quantifier in its scope. Formulas without critical existential quantifiers are exactly formulas of the form~$\forall^* \exists^*$. In what follows, we describe a polynomial time transformation to decrease the number of critical existential quantifiers of a formula. By iterating this procedure, we obtain a formula of the desired form.

Let $\phi$ be a \hyltl formula with at least one critical existential quantifier, i.e., $\phi$ is of the form
\[
 \forall \pi_1 \ldots \forall \pi_n. \exists \pi_{n+1} Q_1 \tau_1 \ldots Q_p \tau_p. \psi\]
 with $n \geq 0$, $p>0$, quantifier-free $\psi$, $Q_i \in \set{\exists, \forall}$, and $Q_i = \forall$ for some $i$.

If $\phi$ has a model $T$, then the quantifier prefix~$\forall \pi_1 \ldots \forall \pi_n. \exists \pi_{n+1}$ of $\phi$ induces a Skolem function~$f$ associating to every $n$-tuple~$(t_1,\ldots,t_n) $ of traces in $ T$ a trace $f(t_1,\ldots,t_n) = t_{n+1}$ in $T$ such that $(T,\Pi) \models Q_1 \tau_1 \ldots Q_p \tau_p. \psi$, where $\Pi$ is the trace assignment with $\Pi(\pi_i) = t_i$ for all $1\leq i \leq n+1$.

We introduce $n+1$ fresh propositions $m^1, \ldots, m^{n+1}$ to mark the traces~$t_1,\ldots,t_{n+1}$ such that $f(t_1,\ldots,t_n) = t_{n+1}$. 
Let $\phi''$ be the conjunction of the sentences
\[ \phi_1 =  \forall \pi_1 \ldots \forall \pi_n. \exists \pi_{n+1}. \F\big(\bigwedge_{i=1}^{n+1}m^i_{\pi_i}\big),  \]
which states the existence of the Skolem function, and
\[ \phi_2 =  \forall \pi_1 \ldots \forall \pi_n. \forall \pi_{n+1}.  Q_1 \tau_1 \ldots Q_p \tau_p. \F\big(\bigwedge_{i=1}^{n+1}m^i_{\pi_i}\big) \Rightarrow \psi,\]
which ensures that the Skolem function is correct. Note that we replaced the existential quantifier on $\pi{n+1}$ by a universal one.
Note that the conjunction of those sentences can be turned into a \hyltl sentence by renaming the trace variables $\pi_i$ in the first one into $\pi'_i$ and then merging them in the following way:
\[\forall \pi'_1 \ldots \forall \pi'_n. \forall \pi_1 \ldots \forall \pi_n. \forall \pi_{n+1}.  Q_1 \tau_1 \ldots Q_p \tau_p. \exists \pi'_{n+1}. \F\big(\bigwedge_{i=1}^{n+1}m^i_{\pi'_i}\big) \land \big[\F\big(\bigwedge_{i=1}^{n+1}m^i_{\pi_i}\big) \Rightarrow \psi\big].\] This formula has one less critical existential quantifier than $\phi$.
Hence, it only remains to prove that $\phi$ and $\phi''$ are equisatisfiable. 

Suppose $\phi$ is satisfiable. Then, it has a countable (possibly finite) model~$T$~\cite{DBLP:conf/stacs/Finkbeiner017}. Thus, for all $t_1,\ldots,t_n \in T$, there exists $t_{n+1} \in T$ with $(T,[(\pi_i \rightarrow t_i)_{1\le i \le n+1}]) \models Q_1 \tau_1, \ldots, Q_p \tau_p, \psi$. Thus, for all $t_1,\ldots,t_n \in T$ there exists a trace assignment~$\Pi_{t_1,\ldots,t_n}$ over $T$ with domain~$\set{\pi_1, \ldots, \pi_n,\pi_{n+1} }$ such that
\begin{itemize}
    \item $(T,\Pi_{t_1,\ldots,t_n}) \models Q_1 \tau_1, \ldots, Q_p \tau_p, \psi$, and
    \item $\Pi_{t_1,\ldots,t_n}(\pi_i) = t_i$ for all $1 \leq i \leq n$.
\end{itemize}

Also, fix an injection~$g \colon T^{n+1} \rightarrow \nats$ (which exists as $T$ is countable).
Now, let $T'' = \set{t'' \mid t \in T} \subseteq (\pow{\ap \cup \set{m^i \mid 1 \le i \le n+1}})^\omega$ where $t''$ is obtained from $t$ as follows:
\begin{itemize}

    \item The projection of $t''$ to $\ap$ is $t$, and
    
    \item for all $1 \leq i \leq n+1$: $m^i \in t''(j)$ if and only if there exist~$t_1, \ldots, t_{n+1}$ such that $t = t_i$, $j = g(t_1,...,t_{n+1})$, and $t_{n+1} = \Pi_{t_1,\ldots,t_n}(\pi_{n+1})$, i.e., we encode the Skolem function using the marks~$m^i$ at the position given by the injection~$g$ to each $(n+1)$-tuple of traces. 
\end{itemize}

We now prove that $T''$ is a model of $\phi''$.
To this end, fix some $t''_1,\ldots,t''_{n+1} \in T''$, and let $\Pi''=[(\pi_i \rightarrow t_i'')_{1 \le i \le n+1}]$. There exist $t_1,\ldots,t_{n+1} \in T$ such that $t_i$ is the projection of $t''_i$ to $\ap$ for all $i$. We denote $\Pi_{t_1,\ldots,t_n}$ by $\Pi$. If $\Pi(\pi_{n+1}) \neq t_{n+1}$ then $m^i \notin t''_i(g(t_1,...,t_n,t_{n+1}))$ for some $i$. 
By injectivity of $g$ and by definition of $T''$, there does not exist any other $j \in \nats$ such that $m^i \in t''_i(j)$ for all $1 \leq i \leq n+1$. Thus, $(T'',\Pi'') \not\models\F(\bigwedge_{i=1}^{n+1}m^i_{\pi_i})$ i.e., $(T'',\Pi'') \models Q_1 \tau_1 \ldots Q_p \tau_p. \F(\bigwedge_{i=1}^{n+1}m^i_{\pi_i}) \Rightarrow \psi$.

If $\Pi(\pi_{n+1}) = t_{n+1}$, i.e., $\Pi(\pi_{i}) = t_{i}$ for all $i$, then by definition of $\Pi = \Pi_{t_1,\ldots,t_n}$, $(T,\Pi) \models Q_1 \tau_1 \ldots Q_p \tau_p. \psi$. As $\psi$ only contains propositions from $\ap$ and as $t_i$ is the projection of $t''_i$ to $\ap$ for all $i$, $(T'',\Pi'') \models Q_1 \tau_1 \ldots Q_p \tau_p. \psi$. Thus, $(T'',\Pi'') \models Q_1 \tau_1 \ldots Q_p \tau_p. \F(\bigwedge_{i=1}^{n+1}m^i_{\pi_i}) \Rightarrow \psi$ as well.
As we have picked $t_0, \ldots, t_{n+1}$ arbitrarily, we have shown $T'' \models \phi_2$.

Hence, it remains to consider $\phi_1$. Recall that in this formula, $\pi_{n+1}$ is quantified existentially. 
Thus, fix some $t''_1, \ldots, t''_n$ in $T''$ and let $t_1, \ldots, t_n$ be their projections to $\ap$. 
As argued above, there is a trace~$t_{n+1}$ in $T$ such that $\Pi_{t_1,\ldots,t_n}(\pi_{n+1})=t_{n+1}$. By definition of $g$ and $T''$, $m^i \in t''_i(g(t_1,\ldots,t_{n+1}))$ for all $i$. Therefore, $(T'',[(\pi_i \rightarrow t''_i)_{1 \leq i \leq n+1}]) \models \F(\bigwedge_{i=1}^{n+1}m^i_{\pi_i})$.
Thus, $T''$ is also a model of $\phi_1$.
Altogether, $T''$ is a model of $\phi''$. 

Conversely, suppose $\phi''$ has a model $T''$. For all $t''\in T''$, let $t$ be its projection to $\ap$. Further, let $T=\set{t \mid t'' \in T''}$. 
Fix some $t_1,\ldots,t_n \in T$. Let $t''_i \in T''$ be such that $t_i$ is the projection of $t''_i$ to $\ap$ for all $i$. As $T''$ models $\phi''$, there exists $t''_{n+1} \in T''$ such that $(T'',[(\pi_i \rightarrow t''_i)_{1 \le i \le n+1}]) \models \F(\bigwedge_{i=1}^{n+1}m^i_{\pi_i})$.

Also, as $T''$ models $\phi''$, $(T'',[(\pi_i \rightarrow t'_i)_{1 \le i \le n+1}]) \models Q_1 \tau_1 \ldots Q_p \tau_p. \F(\bigwedge_{i=1}^{n+1}m^i_{\pi_i}) \Rightarrow \psi$. 
Now, as $\psi$ only contains variables of $\ap$, $(T,[(\pi_i \rightarrow t_i)_{1 \le i \le n+1}]) \models Q_1 \tau_1 \ldots Q_p \tau_p. \psi$.

As we have picked~$t_1,\ldots,t_n$ arbitrarily, we have shown that $T$ is a model of $\phi$.
\qed
\end{proof}

The construction presented in the proof of Lemma~\ref{smallAD} may increase the number of universally quantified variables, but we can decrease that number to two without increasing the temporal or alternation depth.
This step also completes the proof of Theorem~\ref{smalldepth}

This transformation can be achieved by merging several traces into one. To this end, we increase the set of atomic propositions by considering as new atomic propositions tuples of the previous atomic propositions, i.e., one trace now encodes a tuple of traces. However we cannot decrease the number of universal quantifiers below two this way, as we need two universal quantifiers to ensure that every possible combination of traces is represented in the model, i.e., any model of the resulting formula is the set of mergings of traces of another model.

\begin{lemma}\label{fewvar}
For every \hyltl sentence $\phi$ of the form $\forall^*\exists^* \psi$ with quantifier-free $\psi$, one can compute in polynomial time an equisatisfiable sentence $\phi'$ of the form $\forall^2 \exists^* \psi'$ where $\psi'$ is quantifier-free and $\TD(\phi') = \max(\TD(\phi),1)$.
\end{lemma}

\begin{proof}
Let $\phi = \forall \pi_1\ldots\forall \pi_n \exists \pi_{n+1}\ldots\exists \pi_{n+n'} \psi$ with quantifier-free $\psi$, i.e., $\phi$ starts with $n$ universal quantifiers. 
We only consider the case~$n>2$, as otherwise we can just choose $\phi' = \phi$.

We construct a sentence~$\phi'$ over $\ap\times \set{1,\ldots,n}$, i.e., a single trace  represents a tuple of $n$ traces over $\ap$. 
The two universally quantified variables are employed to ensure that every possible tuple is represented by a trace: Given two traces, we state the existence of $n^2$ other traces that are shufflings of those two. We also simulate the universal quantifiers of $\phi$ using the first of the two universal quantifiers in $\phi'$: instead of quantifying $n$ times a single trace, we quantify an $n$-tuple of traces once.

Formally, we define 
\[\phi' = \forall \pi. \forall \pi'. \exists \pi_{n+1}\ldots \exists \pi_{n+n'}. \exists \tau_{(1,1)}\exists \tau_{(1,2)}\ldots\exists \tau_{(n-1,n)}\exists \tau_{(n,n)}. \psi_1 \land \psi_2\] 
where 
\[\psi_1 = \bigwedge_{i_1=1}^n \bigwedge_{i_2 = 1}^n \bigwedge_{a \in \ap} \G\Big[\big((a,i_1)_{\tau_{(i_1,i_2)}} \Leftrightarrow (a,i_2)_{\pi'}\big) \land \bigwedge_{j\neq i_1} \big( (a,j)_{\tau_{(i_1,i_2)}} \Leftrightarrow (a,j)_\pi\big)\Big]
\] expresses that $\pi$ takes all the tuples of traces of the model as values: For any two tuples $(t_1,\ldots,t_n)$, $(t'_1,\ldots,t'_n)$ and for all $1 \leq i_1,i_2 \leq n$ we make sure that the trace $(t_1,\ldots,t_{i_1-1},t'_{i_2},t_{i_1+1},\ldots,t_n)$ also appears. This property implies that all tuples of traces are taken into account when $\pi$ is ranging over them.

Furthermore, $\psi_2$ is obtained from $\psi$ by replacing each atomic formula~$a_{\pi_j}$ for $1 \le j \le n$ with  $(a,j)_\pi$ and every $a_{\pi_j}$ for $n+1 \le j \le n+n'$ by $(a,1)_{\pi_j}$, i.e., we identify the universally quantified traces with the components of a tuple of traces assigned to $\pi$ in $\psi_2$ and the existentially ones with the first components of the tuples assigned to the $\pi_i$ in $\psi_2$.

Suppose $\phi$ has a model $T$. Then, for all $t_1,\ldots,t_n \in T$ let $\merge(t_1,\ldots,t_n)$ be the trace over $\ap \times \set{1,\ldots,n}$ such that $\merge(t_1,\ldots,t_n)(j) = \bigcup_{i=1}^n t_i(j)\times\set{i}$ for all $j \in \nats$. Let $\merge(T) = \set{\merge(t_1,\ldots,t_n) \mid t_1,\ldots,t_n \in T}$. Note that there is a bijection between $n$-tuples over traces from $T$ and traces in $\merge(T)$.

Let $t' \in \merge(T)$, there exist $t_1,\ldots,t_n \in T$ such that $t' = \merge(t_1,\ldots,t_n)$. As $T$ is a model of $\phi$, there exist $t_{n+1},\ldots,t_{n+n'}$ such that $(T, [(\pi_i \rightarrow t_i)_{1 \le i \le n+n'}]) \models \psi$.
Further, let $t'_1,\ldots,t'_n \in T$ be arbitrary. Then, one can check that the trace assignment associating 
\begin{itemize}

    \item $\merge(t_1,\ldots,t_n)$ to $\pi$,
    
    \item $\merge(t'_1,\ldots,t'_n)$ to $\pi'$,
    
    \item $\merge(t_i,\ldots,t_i)$ to $\pi_i$ with $n+1 \le i \le n+n'$, and

    \item $\merge(t_1,\ldots,t_{i_1-1},t'_{i_2},t_{i_1+1},\ldots,t_n)$ to $\tau_{(i_1,i_2)}$
    
\end{itemize}
satisfies both $\psi_1$ and $\psi_2$. Thus, as we have picked the $t_i$ and $t_i'$ for $1 \le i \le n$ arbitrarily, we conclude that $\merge(T)$ is a model of $\phi'$.

Conversely, suppose $\phi'$ has a model $T'$. For all $t' \in T'$, let $t$ be the trace over $\ap$ such that $t(j) = \set{a \mid (a,1) \in t'(j)}$  for all $j \in \nats$, and let $T = \set{t \mid t' \in T'}$.
As $T'$ satisfies $\phi'$, for all $t'_1,t'_2 \in T$, for all $1\leq i_1,i_2 \leq n$ there exist $t'_{i_1,i_2}$ which is identical to $t'_1$ on all $\ap \times \set{i}$ except for $\ap \times \set{i_1}$, for which it is identical to $t'_2$ on $\ap \times \set{i_2}$. Because of this property, one can see that $T'$ is equal to $\set{\merge(t_1,\ldots,t_n) \mid t_1,\ldots,t_n \in T}$.

Now, an induction on the construction of $\psi_2$ shows that, as $T'$ satisfies $\phi'$, $T$ satisfies~$\phi$.
\qed
\end{proof}

Thus, $\forall^2\exists^*$ formulas with temporal depth two capture the complete complexity of the satisfiability problem for \hyltl. 
As the latter problem is undecidable and as all reduction presented above are effective, we immediately obtain that  satisfiability for $\forall^2\exists^*$ formulas with temporal depth two is also undecidable.

As alluded to above, an even stronger result can be obtained by strengthening the proof of Hahn and Finkbeiner for $\forall \exists^*$ formulas to only use temporal depth two. \footnote{Alternatively, one can also obtain a direct reduction from the Turing machine immortality problem~\cite{DBLP:journals/jacm/Hooper66a} to satisfiability of $\forall \exists$ sentences of temporal depth two.}
Thus, only formulas of temporal depth one remain to consider.

Before we start investigating this class let us quickly comment on why we disregard temporal depth zero: every such sentence can easily be turned to an equisatisfiable instance of \qbf, which is known to be solvable in polynomial space.

Thus, it only remains to consider formulas with arbitrary quantifier prefixes, but temporal depth one.
Our main result of this section shows that even this problem is undecidable, even for \hyltlfgone formulas with alternation depth one. Due to the restriction on the temporal depth, our encoding of a Minsky machine is more complicated than it would be with arbitrary temporal depth.

\begin{theorem}\label{SatFX1}
The following problem is undecidable: 
Given a $\forall^2\exists^*$ \hyltlfgone sentence~$\phi$, is $\phi$ satisfiable?
\end{theorem}

\begin{proof}
We reduce from the (non)-halting problem for 2-counter Minsky machines. Recall that such a machine can be seen as a tuple $\M = (Q,\Delta,q_0)$ where $Q$ is a finite set of states, $q_0 \in Q$ an initial state, and $\Delta \subseteq Q \times \set{1,2} \times OP \times Q $ a set of transition rules, where $OP = \set{\incop , \decop , \testop}$. A configuration of $\M$ is an element of $Q \times \nats \times \nats$. For all $n, n' \in \nats$, $op \in OP$ we write $n \xrightarrow{op} n'$ if:
\begin{itemize}
    \item $op$ is $\incop$ and $n' = n+1$.
    \item $op$ is $\decop$ and $n' = n-1$ (note that this operation is only applicable if $n>0$).
    \item $op$ is $\testop$ and $n' = n = 0$.
\end{itemize}
There is a transition from $(q,n_1,n_2)$ to $(q',n'_1,n'_2)$ if and only if there exists $i \in \set{1,2}$ and $op\in OP$ such that $(q,i,op,q') \in \Delta$, $n_{3-i} = n_{3-i}'$ and $n_i \xrightarrow{op} n'_i$. Given such a machine, it is undecidable whether it has an infinite computation $(q_0,0,0) \rightarrow (q_1,n_1^1, n_2^1) \rightarrow (q_2,n_1^2, n_2^2) \rightarrow \cdots$.

Let $\M = (Q,\Delta,q_0)$ be a 2-counter Minsky machine. We use $\ap = Q \cup \set{1,2}$ as atomic propositions. Given $i \in \set{1,2}$, we denote by $\otherone{i}$ the other proposition.
Consider the formula $\psi_1 = \forall \pi. \forall \pi'. \G (1_\pi \Rightarrow 1_{\pi'}) \lor \G (1_{\pi'} \Rightarrow 1_\pi)$. We define $\psi_2$ with $2 \in\ap$ analogously. In the following, we only consider sets of traces that satisfy $\psi_1 \wedge\psi_2$. 

 For each trace $t \in (\pow{\ap})^\omega$ and $i\in\set{1,2}$, we define the \emph{$i$-set} of $t$ as $S_i(t) = \set{j \in \nats \mid i \in t(j)}$. Now fix $T \subseteq (\pow{\ap})^\omega$ that satisfies $\psi_1 \wedge\psi_2$. We define the order~$\leq_i$ on $T$ as follows: for all $t,t' \in T$, $t \leq_i t'$ if and only if $S_i(t) \subseteq S_i(t')$. We write $t <_i t'$ if $S_i(t) \subsetneq S_i(t')$. As $T$ satisfies $\psi_1 \wedge\psi_2$, the $\leq_i$ are total orders on $T$. We also define for all $t \in T$ and $i \in \set{1,2}$, the \emph{rank} of $t$ with respect to $i$ as $ \rk{i}{t} = \size{\set{S_i(t') \mid t' \in T\text{ and } t' <_i t}} $, which may be infinite. Note that if $S_i(t) = \emptyset$ then $\rk{i}{t} = 0$, and that if $S_i(t) = S_i(t')$ then $\rk{i}{t} = \rk{i}{t'}$. 
Finally, as $\leq_i$ is a total order, if we have $t <_i t'$, but there is no $t''$ with $t <_i t'' <_i t'$, then $\rk{i}{t'} = \rk{i}{t} +1$. Note that this holds even when $\rk{i}{t}$ is infinite, assuming $\infty + 1 = \infty$.

We construct a \hyltlfgone formula~$\phi$ that encodes the existence of an infinite computation~$(q_0,0,0) \rightarrow (q_1,n_1^1, n_2^1) \rightarrow (q_2,n_1^2, n_2^2)\rightarrow\cdots$ of $\M$.
In a model $T$ of $\phi$, a configuration $(q,n_1,n_2)$ is encoded by a trace $t$ with $t(0) \cap Q = \set{q}$ and for $i \in \set{1,2}$, $rk_i(t) = n_i$. 
Then, $\phi$ states the existence of an initial trace~$t_0$, representing the configuration $(q_0,0,0)$, as well as the existence of a successor $t'$ encoding $(q', n_1', n_2')$ for each trace $t$ encoding $(q, n_1, n_2)$, i.e., we require $(q, n_1, n_2) \rightarrow (q', n_1', n_2')$.
The latter is witnessed by the existence of a  transition~$(q,i,op,q')$ such that:
\begin{enumerate}
    
    \item $t(0) \cap Q = \set{q}$ and $t'(0) \cap Q = \set{q'}$, i.e., $t$ and $t'$ indeed encode the states of their respective configurations correctly.
    
    \item  For all $j \in \nats$, $\otherone{i} \in t(j)$ if and only if $\otherone{i} \in t'(j)$, i.e. $S_{\otherone{i}}(t) = S_{\otherone{i}}(t')$. Thus, as argued above, $\rk{\otherone{i}}{t} = \rk{\otherone{i}}{t'}$, which implies $n_{\otherone{i}} = n'_{\otherone{i}}$.
    
    \item If $op$ is $\incop$ then $t <_i t'$ and there does not exist any $t''$ such that $t <_i t'' <_i t'$, i.e., $\rk{i}{t'} = \rk{i}{t} +1$, as $\leq_i$ is a total order.
    Then, we have $n_i' = n_i +1$.
    
    \item If $op$ is $\decop$ then $t >_i t'$ and there does not exist any $t''$ such that $t >_i t'' >_i t'$, i.e., $\rk{i}{t'} = \rk{i}{t} -1$, as $\leq_i$ is a total order.
    Then, we have $n_i' = n_i -1$.
        
    \item If $op$ is $\testop$ then for all $j \in \nats$, $i \notin t(j)$  and $i \notin t'(j)$.
    Hence, $S_i(t) = S_i(t')= \emptyset$, i.e., $\rk{i}{t}= \rk{i}{t'} =0$, which implies $n_i = n_i' = 0$.
\end{enumerate}

We encode those conditions in $\phi$, which is the conjunction of the following three sentences
and of $\psi_1 \wedge\psi_2$:
\begin{itemize}

\item $\phi_1 = \forall \pi. \bigwedge_{q \neq q' \in Q} q_\pi \Rightarrow \neg q'_\pi$ expresses that a trace is associated to at most one state.

\item $\phi_2 = \exists \pi_0. {(q_0)}_{\pi_0} \land \G(\neg 1_{\pi_0} \land \neg 2_{\pi_0})$ expresses the existence of a trace representing the initial configuration $(q_0,0,0)$.

\item $\phi_3 = \forall \pi. \exists \pi'. \bigvee_{(q,i,op,q') \in \Delta} q_\pi \land q'_{\pi'} \land \phi_{i,op} \land \G(\otherone{i}_\pi \Leftrightarrow \otherone{i}_{\pi'})$ expresses that all traces have a successor obtained by faithfully simulating a transition of the machine.
\end{itemize}
Here, we use the formulas
\begin{itemize} 
\item $\phi_{1,\incop} = \forall \pi''. \pi <_1 \pi' \land (\pi'' \leq_1 \pi \lor \pi' \leq_1 \pi'')$,  
\item $\phi_{1,\decop} = \forall \pi''. \pi >_1 \pi' \land (\pi'' \geq_1 \pi \lor \pi' \geq_1 \pi'')  $, and 
\item $\phi_{1,\testop} = \G(\neg 1_\pi \land \neg 1_{\pi'}) $,
\end{itemize}
 where
$\pi \leq_1 \pi' = \G(1_\pi \Rightarrow 1_{\pi'}) $ and $\pi <_1 \pi' = \pi \leq_1 \pi' \land \F(\neg 1_{\pi} \land 1_{\pi'}) $. Finally, we define the formulas~$\leq_2$, $<_2$, and $\phi_{2,op}$ analogously.

The sentence~$\phi$ is not in prenex normal form. However, as no quantifier appears in the scope of a temporal operator, it can be put in that form. Further, it is not of the form $\forall^2\exists^*$, but we can apply Lemmas~\ref{smallAD} and \ref{fewvar} to bring it into this form while preserving the temporal depth, which is already one. We claim that $\phi$ is satisfiable if and only if $\M$ has an infinite computation starting in $(q_0,0,0)$.

Suppose $\phi$ is satisfied by a model $T$. The subformulas $\phi_1$ and $\phi_2$ enforce that $T$ contains a trace $t_0$ encoding the initial configuration $(q_0, 0 , 0)$ of $\M$. Further, $\phi_3$ expresses that every trace $t$ encoding a configuration~$(q,n_1,n_2)$ has a successor $t'$ encoding a configuration~$(q', n_1', n_2')$ with $(q,n_1,n_2) \rightarrow (q',n_1',n_2')$.
Thus, there exists an infinite sequence~$t_0,t_1,t_2,\ldots$ of traces encoding an infinite run of $\M$. 

Conversely, suppose $\M$ has an infinite run $(q_0,0,0) \rightarrow (q_1,n_1^1,n_2^1) \rightarrow(q_2,n_1^2,n_2^2)\cdots$, then for all $j$ let $t_j$ be the trace whose projection to $Q$ is $\set{q_j}\emptyset^\omega$, and whose projection to $\set{i}$ is $\set{i}^{n_i^j}\emptyset^\omega$ for $i \in \set{1,2}$.
One can then easily check that $\set{t_j \mid j \in \nats}$ is a model of~$\phi$.
\qed
\end{proof}

Thus, two universal quantifiers before some existential ones and using only $\F$ and $\G$ without nesting yields undecidable satisfiability. Our next result shows that removing one of the two universal quantifiers allows us to recover decidability, even when allowing nested next operators and leading existential quantifiers.

As a first step in the proof, we show that the nested next operators can be eliminated without introducing additional universal quantifiers. This is true, as we are only interested in satisfiability.

\begin{lemma}\label{xelim}
For every $\exists^* \forall \exists^*$ \hyltlfgxone sentence, one can construct in polynomial time an equisatisfiable $\exists^* \forall \exists^* $ \hyltlfgone sentence.
\end{lemma}

\begin{proof}
Let $\phi = \exists \tau_1\ldots\exists \tau_n. \forall \pi . \exists \tau_{n+1}\ldots\exists \tau_{n+n'}. \psi$ be a \hyltlfgxone sentence. We assume w.l.o.g., that $\psi$ is a Boolean combination of formulas of the form $\F \beta$ and $\X^k \beta$ where $\beta$ is a Boolean combination of atomic propositions. Let $d$ be the maximal integer such that $\psi$ contains a subformula of the form $\X^d \beta$ (which is $0$, if there is no $\X$ in $\psi$).

We extend the set of atomic propositions with fresh ones $m^0,\ldots,m^d$ and define
\[\phi' = \exists \tau_0. \exists \tau_1\ldots\exists \tau_n. \forall \pi . \exists \tau_{n+1}\ldots\exists \tau_{n+n'}. \psi'\] (i.e., we just add one existentially quantified variable~$\tau_0$) with 
\[\psi' = \bigwedge_{k=0}^d \Big(\F m^k_{\tau_0}\Big) \land \Big(\F\big(\bigvee_{a \in \ap \cup \set{m^0,\ldots,m^d}} a_\pi \oplus a_{\tau_0}\big) \Rightarrow \psi''\Big)\] where  $\psi''$ is $\psi$ where every maximal subformula of the form $\X^k \beta$ has been replaced with $\G({m^k}_{\tau_0} \Rightarrow \beta )$. Intuitively, we state the existence of a special trace on which the atomic propositions $m^0,\ldots,m^d$ mark some positions which will simulate the positions from $0$ up to $d$ in any order with any multiplicity. As we are dealing with satisfiability, and as the order of the valuations at these positions is irrelevant for the satisfaction of a \hyltlfgone formula, we obtain an equisatisfiable formula. 

If $\phi$ has a model $T$, then let $T' = T \cup \set{t_0}$, with $t_0 = \set{m^0}\ldots\set{m^d}\emptyset^\omega$. Here every $m^k$ appears once, at position $k$, thus, $\X^k \beta$ is satisfied by some assignment over $T$ if and only $\G (m^k \Rightarrow \beta)$ is satisfied by the same assignment over $T'$. Then, one can show by a simple induction over the construction of $\psi'$ that for all $t_1,\ldots,t_{n+n'},t$, $(T,[(\tau_i \rightarrow t_i)_{1\leq i \leq n+n'}, \pi \rightarrow t])$ satisfies $\psi$ if and only if $(T',[(\tau_i \rightarrow t_i)_{0\leq i \leq n+n'}, \pi \rightarrow t])$ satisfies $\psi'$. Since $\phi$ and $\phi'$ have the same quantifier prefix (but for the quantification of $\tau_0$), and as $T$ satisfies $\phi$, $T'$ satisfies $\phi'$.

Conversely, if $\phi'$ has a model $T'$, then let $t'_0 \in T'$ be such that $(T',[\tau_0 \rightarrow t'_0])$ satisfies \[\exists \tau_1\ldots\exists \tau_m. \forall \pi . \exists \tau_{n+1}\ldots\exists \tau_{n+n'}. \psi'.\] In particular $(T',[\tau_0 \rightarrow t'_0])$ satisfies $\bigwedge_{k=0}^d \F m^k_{\tau_0}$. Thus, there exist $j_0,\ldots,j_k \in \nats$ such that $m^k \in t'_0(j_k)$ for all $k$. For all $t' \in T'\setminus \set{t'_0}$ let $t$ be the projection to $\ap$ of $t'(j_0)\cdots t'(j_d)t'$, and let $T = \set{t \mid t'\in T'\setminus \set{t'_0}}$. Thus, if a trace assignment over $T'$ satisfies $\G (m^k_{\tau_0} \Rightarrow \beta )$ then the associated trace assignment over $T$ satisfies $\X^k\beta$. Furthermore, if a trace assignment over $T'$ satisfies $\F \beta$ (respectively $\G \beta$) then so does the associated trace assignment over $T$ (we have not modified the set of combinations of valuations that eventually appear, just their order). Hence, as $\psi$ can be seen as a positive Boolean combination of such formulas, an induction over $\phi$ shows that, as $\phi'$ is satisfiable, so is $\phi$.
\qed
\end{proof}

Now, we are ready to prove our main decidability result in this section. Note that we do not claim a matching lower bound here. 
We comment on this gap in the conclusion.

\begin{theorem}\label{sat1forall}
The following problem is in \ntwoexpt: 
Given a \hyltlfgxone sentence~$\phi$ of the form~$\exists^* \forall \exists^*$, is $\phi$ satisfiable?
\end{theorem}

\begin{proof}
Let $\phi = \exists \tau_1\ldots \tau_n. \forall \pi. \exists \tau_{n+1}\ldots\exists \tau_{n+n'}. \psi$ be a \hyltlfgxone sentence with quantifier-free~$\psi$. Due to Lemma \ref{xelim}, it is enough to consider the case where $\psi$ is a Boolean combination of formulas of the form $\F \beta$ for a Boolean combination $\beta$ of atomic propositions.

To every tuple~$(t_1,\ldots,t_k)$ of traces~$t_i \in (\pow{\ap})^\omega$, we associate a finite set of tuples of valuations $V(t_1,\ldots,t_k) = \set{(t_1(j),\ldots,t_k(j)) \mid j \in \nats} \subseteq (\pow{\ap})^k $, i.e., the set all the tuples of valuations that appear eventually.
The cardinality of $V(t_1,\ldots,t_k)$ is at most $\pow{k\size{\ap}}$. 

Let $\beta$ be a Boolean combination of atomic propositions over trace variables $\pi_1,\ldots, \pi_k$. Then, a trace assignment $[(\pi_i \rightarrow t_i)_{1 \leq i \leq k}]$ satisfies $\F \beta$ if and only if there exists $j \in \nats$ such that $\beta$ is satisfied at position $j$ of $(t_1,\ldots,t_k)$, i.e., there exists $(v_1,\ldots,v_k) \in V(t_1,\ldots,t_k)$ such that $(v_1,\ldots,v_k)$ satisfies $\beta$ (in the sense that any trace assignment $\Pi$ such that $\Pi(\pi_i)(0) = v_i$ for all $i$ satisfies $\beta$). Intuitively, we abstract a tuple of traces into a finite set of tuples of valuations, and then abstract a model as a set of such finite representations. Then, we show that satisfiability can be decided using such abstractions. 

As a consequence, whether a given trace assignment $[(\pi_i \rightarrow t_i)_{1 \le i \le k}]$ satisfies a given Boolean combination $\psi$ of formulas~$\F \beta$ only depends on $V(t_1,\ldots,t_k)$, and given $V \subseteq (\pow{\ap})^k$, one can check in polynomial time whether a trace assignment yielding $V$ satisfies $\psi$. If it is the case, we say that $V$ satisfies $\psi$.

To check the satisfiability of $\phi$, we start by nondeterministically guessing a set $S \subseteq \pow{(\pow{\ap})^{n+n'+1}}$ of sets of $(n+n'+1)$-tuples of valuations. This set is supposed to represent a model of $\phi$. The $n$ first valuations represent the fixed values assigned to $\tau_1,\ldots,\tau_n$. The $(n+1)$-th represents the valuation of the universally quantified variable. Thus, for every trace of the model there must exist a tuple in which that trace is represented at position $n+1$. The valuations of positions $n+2$ to $n+n'$ have to be such that $\phi$ is satisfied by all tuples.

Thus, we check the following requirements:
\begin{enumerate}

    \item\label{sat1forallconsistent} For all $V_1,V_2 \in S$, $\set{(v_1,\ldots,v_n) \mid (v_1,\ldots,v_{n+n'+1}) \in V_1} $ is equal to $ \set{(v_1,\ldots,v_n) \mid (v_1,\ldots,v_{n+n'+1}) \in V_2}$: The set of values taken by the traces assigned to $\tau_1,\ldots,\tau_n$ cannot depend on the values of the other variables. Thus, we ensure that these values are fixed in the guessed model.
    
    \item\label{sat1forallalltraces} For all $V \in S$ and $1 \leq i \leq n+n'+1$, there exists $V' \in S$ such that $\set{(v_1,\ldots,v_n,v_i) \mid (v_1,\ldots,v_{n+n'+1}) \in V} = \set{(v_1,\ldots,v_{n+1}) \mid (v_1,\ldots,v_{n+n'+1}) \in V'}$. All the values taken by the existentially quantified variables have to be taken by the universally quantified one as well. 

   \item\label{sat1forallsatisfy} For all $V \in S$, $V$ satisfies $\psi$.

\end{enumerate}

If all requirements are satisfied, we accept, otherwise we reject. This procedure requires nondeterministic doubly-exponential time as $\size{S} \leq \pow{\pow{\size{\ap}+n+n'+1}}$.

Suppose $\phi$ is satisfiable and fix a model~$T$. There exist $t_1,\ldots,t_n \in T$ such that $(T,[(\tau_i \rightarrow t_i)_{1 \leq i \leq n}]) \models \forall \pi \exists \tau_{n+1}\ldots\exists \tau_{n+n'}. \psi$. Furthermore, for a fixed $t \in T$ there exist $t_{n+1},\ldots,t_{n+n'} \in T$ such that $(T,[(\tau_i \rightarrow t_i)_{1\leq i \leq n+n'},\pi \rightarrow t]) \models \psi$. Let $V^*(t) = \set{(t_1(j),\ldots, t_n(j),t(j),t_{n+1}(j),\ldots,t_{n+n'}(j)) \mid j \in \nats}$. 

Now, one can easily check that Requirements \ref{sat1forallconsistent}, \ref{sat1forallalltraces}, and \ref{sat1forallsatisfy} are satisfied by  $ \set{V^*(t) \mid t \in T}$. Thus, the algorithm accepts $\phi$.

Conversely, suppose the algorithm accepts $\phi$.
Then, there exists some $S$ satisfying all three requirements above. We construct from $S$ a model $T$ of $\phi$.

Let $t_1,\ldots,t_n$ be traces such that for all $V \in S$, $\set{(v_1,\ldots,v_n) \mid (v_1,\ldots,v_{n+n'+1}) \in V} = V(t_1,\ldots,t_n)$, and for all $(v_1,\ldots,v_{n+n'+1}) \in V$, $(v_1,\ldots,v_n) = (t_1(j), \ldots, t_n(j))$ for infinitely many $j$, i.e., each of the valuations appears infinitely often in the traces. Those traces can be constructed due to Requirement~\ref{sat1forallconsistent}.

Let $T_0 = \set{t_1,\ldots,t_n}$. For all $\ell \in \nats$ we construct $T_\ell$ by induction on $\ell \in\nats$, while maintaining the following two invariants:
\begin{enumerate}
    \item \label{sat1forallaccurate} For all $t \in T_\ell$ there exists $V \in S$ such that $V(t_1,\ldots,t_n,t) $ is equal to $ \set{(v_1,\ldots,v_{n+1}) \mid (v_1,\ldots,v_{n+n'+1}) \in V}$, and for all $(v_1,\ldots,v_{n+n'+1}) \in V$, $(v_1,\ldots,v_{n+1}) $ is equal to $ (t_1(j),\ldots,t_n(j),t(j))$ for infinitely many $j$, where the $t_i$ are the traces in $T_0$.
    
    \item \label{sat1forallimage} If $\ell > 0$ then for every $t \in T_{\ell-1}$, there exist traces $t_{n+1},\ldots,t_{n+n'} \in T_\ell$ such that $[(\tau_i \rightarrow t_i)_{1 \leq i \leq n+n'},\pi \rightarrow t] \models \psi$.
    
\end{enumerate}

By Requirement \ref{sat1forallalltraces} and by construction, $T_0$ satisfies Invariant \ref{sat1forallaccurate}, and it clearly satisfies Invariant \ref{sat1forallimage}.
Let $\ell \in \nats$, suppose $T_\ell$ has been constructed, and that it satisfies Invariants \ref{sat1forallaccurate} and \ref{sat1forallimage}. By Invariant \ref{sat1forallaccurate}, for all $t \in T_{\ell}$ we can construct traces $t_{n+1},\ldots,t_{n+n'}$ such that $V(t_1,\ldots,t_{n},t, t_{n+1}, \ldots, t_{n+n'}) \in S$ and for all $(v_1,\ldots,v_{n},v, v_{n+1}, \ldots, v_{n+n'}) \in V(t_1,\ldots,t_{n},t, t_{n+1}, \ldots, t_{n+n'})$, it is the case that $(v_1,\ldots,v_{n},v, v_{n+1}, \ldots, v_{n+n'}) $ is equal to \newline$ (t_1(j),\ldots,t_{n}(j),t(j), t_{n+1}(j), \ldots, t_{n+n'}(j))$ for infinitely many $j$ (as all the $(v_1,\ldots,v_{n},v)$ appear infinitely many times in $(t_1,\ldots,t_n,t)$ by Invariant~\ref{sat1forallaccurate}). Let $I(t) = \set{t_{n+1},\ldots,t_{n+n'}}$. Let $T_{\ell+1} = \bigcup_{t \in T}I(t)$, which satisfies Invariant \ref{sat1forallaccurate} by Requirement \ref{sat1forallalltraces}. It also satisfies Invariant \ref{sat1forallimage} by definition. Furthermore, by Requirement \ref{sat1forallsatisfy}, $V(t_1,\ldots,t_{n},t, t_{n+1}, \ldots, t_{n+n'})$ satisfies $\psi$.

Finally, let $T = \bigcup_{\ell \in \nats} T_\ell$ and let $t \in T$. Then, there exists an $\ell$ such that $t \in T_\ell$. Thus, there also exist $t_{n+1},\ldots,t_{n+n'} \in T_{\ell+1}$ such that $[(\tau_i \rightarrow t_i)_{1 \leq i \leq n+n'}, \pi \rightarrow t]$ satisfies $\psi$. Therefore, $T$ satisfies $\phi$.
\qed
\end{proof}

Recall that satisfiability of $\exists^*\forall^*$ formulas is $\exps$-complete~\cite{DBLP:conf/concur/FinkbeinerH16}. The proof of Finkbeiner and Hahn can be slightly adapted to produce a formula of temporal depth two: their approach states the existence of a trace representing a sequence of configurations of an exponential-space bounded Turing machine. The only difficulty that can arise in expressing the correctness of the run described by that trace is relating a position of one of the configurations to the neighbouring positions in the next configuration (in order to simulate the movement of the head). One may then require to combine an until and a next in order to express this requirement, in the scope of an always expressing that it holds for every position. This nesting can be removed by adding a fresh proposition~$p$ that is satisfied on all positions of the first configuration, on none of the second one, and so on, i.e., its truth value alternates between the configurations. One can then express the previous requirement with a single until in the scope of an always, yielding temporal depth two.

Our next result shows that one obtains slightly better complexity when restricting the temporal depth of formulas to one. 

\begin{theorem}\label{nexp}
The following problem is \nexpt-complete: Given an $\exists^*\forall^*$ \hyltl sentence~$\phi$ with temporal depth one, is $\phi$ satisfiable?
\end{theorem}

\begin{proof}
This proof is an adaptation of the proof by Finkbeiner and Hahn showing that the problem is \exps-complete for sentences of arbitrary temporal depth~\cite{DBLP:conf/concur/FinkbeinerH16}. For the upper bound, we use the transformation of Finkbeiner and Hahn to turn the initial sentence into an equisatisfiable exponentially larger one of the form $\exists^* \psi'$ without increasing the temporal depth, by replacing universal quantifiers with conjunctions.
Then, we turn that sentence into an equisatisfiable \ltl formula of the same size and the same temporal depth by merging traces, again due to Finkbeiner and Hahn. Finally, we apply a result of Demri and Schnoebelen stating that satisfiability of \ltl formulas of temporal depth one is in \np~\cite{DBLP:journals/iandc/DemriS02}. Altogether, we obtain the desired algorithm.

For the lower bound, we again adapt ideas of Finkbeiner and Hahn. We reduce from the following \nexpt-complete problem: Given a nondeterministic Turing machine $\M$ and an integer $n$ (in unary), does $\M$ accept the empty word in time at most $2^n$? Given such a machine $\M = (Q,\Sigma,\Delta,q_0,q_f)$, we construct a \hyltlfgone formula $\phi$ of the form $\exists^* \forall^* \psi$ such that $\phi$ is satisfiable if and only if $\M$ accepts the empty word in less than $2^n$ steps. We can assume that there is a loop on the final state, i.e., the machine accepts in at most $2^n$ steps if and only if it accepts in exactly $2^n$ steps. We also assume $\Sigma$ to contain a blank letter $B$ meaning that nothing has been written at that position yet. In particular, the initial tape contains only $B$.

In our encoding of a run of $\M$, the valuation at some position encodes the content of one cell at one moment in time.
Thus, we use the set~$Q \cup \Sigma \cup \set{\spaceprop_1,\ldots,\spaceprop_n} \cup \set{\timeprop_1,\ldots,\timeprop_n} \cup \set{h} \cup \set{0,1}$ of atomic propositions, where $h$, $\spaceprop_i$, and $\timeprop_i$ are fresh atomic propositions representing respectively whether the head is at the encoded cell, the position on the tape and time (the two latter in binary). For all $0 \leq j < 2^n$ let $time(j)$ and $space(j)$ be the valuations over the $\timeprop_i$ and $\spaceprop_i$ representing $j$.

To encode the existence of an accepting run of $\M$, we require the existence of a trace $t$ encoding such a run in the following sense:
\begin{itemize}

    \item Every valuation of $\timeprop_1,\ldots,\timeprop_k,\spaceprop_1,\ldots,\spaceprop_k$ appears eventually on $t$.
    
    \item Two positions with the same valuation of the $\timeprop_i$ and $\spaceprop_i$ also have the same valuation of the other atomic propositions.
    
    \item For all $j \in \nats$, $t(j)$ contains exactly one letter from $\Sigma$ and one state from $Q$.
    
    \item There exists a run of $\M$ such that for all $j_t, j_s$, if at the $j_t$-th step the machine is in state~$q$ and has an $a$ at position $j_s$, then for all $m \in \nats$, if $t(m) \cap \set{\timeprop_i \mid 1 \leq i \leq n} = time(j_t)$ and $t(m) \cap \set{\spaceprop_i \mid 1 \leq i \leq n} = space(j_s)$ then $a, q \in t(m)$ and $h \in t(m)$ if and only if the head is at position $j_s$ at the $j_t$-th step of the run, i.e., every trace faithfully represents the content of the cell at the step of the run and the position on the tape it is assigned to.
\end{itemize}

Now, we define $\phi$ as 
\begin{multline*}
\exists \pi. \exists \pi_0. \exists \pi_1 . \forall \tau_1 \ldots \forall \tau_{n}. \forall \tau'_1 \ldots \forall \tau'_{n}. \forall \sigma_1 \ldots \forall \sigma_{n}. \forall \sigma'_1 \ldots \sigma'_{n}.\\
\psi_{0/1} \land [\psi_{assign} \Rightarrow (\psi_{halt} \land \psi_{init} \land \psi_{cons} \land \psi_{complete} \land \bigvee_{(q,q',a,d) \in \Delta} \psi_{trans}(q,q',a,a',d))].
\end{multline*}
We explain the role of $\phi$'s subformulas below, but leave the construction of the formulas to the reader in order to improve readability.
We just remark that all can be constructed with temporal depth one.

\begin{itemize}
    \item $\psi_{0/1}$ expresses that $\pi_0$ takes value $\set{0}^\omega$ and that $\pi_1$ takes value $\set{1}^\omega$.

    \item $\psi_{assign}$ ensures that all the traces assigned to the~$\tau_i,\tau'_i,\sigma_i, \sigma'_i$ are equal to either the trace assigned to $\pi_0$ or the one assigned to $\pi_1$. The values of those traces represent two valuations of the $\timeprop_i$ ($\tau_i$ and $\tau'_i$) and two of the $\spaceprop_i$ ($\sigma_i$ and $\sigma'_i$). 
    
    \item $\psi_{halt}$ ensures that there exists a position in the trace assigned to $\pi$ containing $q_f$.
    
    \item $\psi_{init}$ ensures that all the positions having the $\timeprop_i$ representing $0$ contain a $B$ and a $q_0$, and that positions having the $\timeprop_i$ and $\spaceprop_i$ representing $0$ contain an $h$, i.e., the initial configuration is encoded correctly.
    
    \item $\psi_{cons}$ checks that all positions of the trace assigned to $\pi$ in which the valuation of the time and space propositions is the one represented by the traces assigned to the $\tau_i$ and $\sigma_i$ have the same valuation over the other variables. It also checks that there is exactly one letter and one state at each position of the trace assigned to $\pi$, and that if two positions in the trace assigned to $\pi$ have the same time valuations (represented by the trace assigned to $\tau_i$) but not the same space valuations, then at most one of the two satisfies the head proposition $h$, and they have the same state.
    
    \item $\psi_{complete}$ expresses that there exists a position in $\pi$ matching the $\timeprop_i$ and $\spaceprop_i$ with the $\tau_i$ and $\sigma_i$, i.e., that the complete computation tableau is encoded.
    
    \item $\psi_{trans}(q,q',a,a',d)$ checks that the position of the head, the state and the change in the letters are consistent with that transition. For instance, if the $\sigma_i$ and $\sigma_i'$ represent the same valuation and the $\tau'_i$ represent the successor of the $\tau_i$ and the positions in $\pi$ matching the $\tau_i$ and $\sigma_i$ contain $h$, then $\psi_{trans}$ ensures that the positions matching $\tau_i$ and $\sigma_i$ contain $a$ and $q$ and the ones matching $\tau'_i$ and $\sigma'_i$ contain $a'$ and $q'$.
\end{itemize}

One can then see that this formula is satisfiable if and only if $\M$ accepts in time at most~$2^n$, as it explicitly describes the existence of an accepting run of $\M$.
\qed
\end{proof}

We conclude by considering the satisfiability problem for \hyltlfgone with arbitrary quantifier prefixes, but restricted to models induced by finite-state systems. The undecidability of satisfiability for arbitrary formulas over finite-state systems can be easily inferred from the proof of undecidability of satisfiability of Finkbeiner and Hahn, as the formulas they construct, if satisfiable, have a finite and ultimately periodic model, which is therefore representable by a finite-state system.
For formulas of \hyltlfgone, we leave decidability open, but prove intractability.

\begin{theorem}\label{SatFX1K}
The following problem is $\tower$-hard: Given a \hyltlfgone sentence~$\phi$, does $\phi$ have a model~$T(\K)$ for some Kripke structure~$\K$?
\end{theorem}

\begin{proof}
Let $\M$ be a deterministic Turing machine $(Q,\Sigma, \delta, q_0,q_f)$, let $w \in \Sigma^*$, and let $n = \size{w}$. We construct a sentence~$\phi$ that is satisfiable by a Kripke structure if and only if $\M$ accepts $w$ in time $\twr(n)$, where $\twr(0)=2$ and $\twr(k+1) = \pow{\twr(k)}$ for all $k \in \nats$. This problem can easily be shown \tower-complete under elementary reductions.

To this end, we use the set \[\ap = \Sigma \cup Q \cup \set{\rankprop_0,\ldots,\rankprop_n} \cup \set{h, 0, 1, 2, \inmark, \nimark, \tmark, \smark, \ltmark , \gtmark}\] of propositions.
There are two types of traces: Type one traces, marked by the atomic proposition $1$, allow us to count up to $\twr(n)$, while type two traces simulate an accepting run of $\M$ on $w$. Each type two trace encodes one cell of the tape of $\M$ at one step of the computation.
We use the formula $\psi_{type} = \forall \pi. 1_\pi \oplus 2_\pi$ to ensure that the traces are split between the two types.

Each type one trace is assigned a \emph{rank} between $0$ and $n$, which is indicated by the propositions~$\rankprop_0,\ldots,\rankprop_n$. A trace of rank $r+1$ encodes a set of traces of rank $r$. There are two traces of rank zero, which represent zero and one. Thus, we are able to generate $\twr(r)$ traces of rank $r$ for all $0 \leq r \leq n$. For instance, at rank one we have four traces, representing all possible subsets of the two rank zero traces.

Then, we construct an order on those traces and use traces of rank $n$ to simulate the Turing machine. More precisely, each trace of type two is associated to two traces of type one, one for space and one for time, which encode its position on the tape and the step of the computation it describes.

We require that for all type one traces $t$ there is exactly one $r$ such that $\rankprop_r \in t(0)$, and we then say that $t$ is of rank $r$. 
If $t$ is of rank zero then it is either $t_0 = \set{1,r_0,0}^\omega$ or $t_1 = \set{1,r_0}^\omega$. These properties can easily be expressed in \hyltlfgone.

A type one trace $t$ of rank $r+1$ is said to contain a type one trace $t'$ of rank $r$ (denoted by $t' \in t$) if and only if there exists $j \in \nats$ such that $\inmark \in t'(j)$ and $\nimark \in t(j)$.  To every trace $t$ of rank~$r$ we associate a set $S(t)$, which is an element of the set obtained by iterating the power set operator on $\set{\set{0},\set{1}}$ $r$ times:
\begin{itemize}
    \item $S(t_0) = \set{0}$ and $S(t_1)=\set{1}$.
    
    \item Else, $S(t) =\set{S(t') \mid t' \in t}$.
\end{itemize}
The following formula is satisfied by $[\pi \rightarrow t, \pi' \rightarrow t']$ if and only if $t \in t'$: \[\psi_{in}(\pi,\pi') = 1_\pi \land 1_{\pi'} \land \F( {\nimark}_{\pi'} \land {\inmark}_{\pi} ) \land \bigvee_{r=0}^{n-1} ({\rankprop_r})_{\pi'} \land ({\rankprop_{r+1}})_\pi.\]

Furthermore, we define a strict order $\lord$ over traces of equal rank: $t_0 \lord t_1$ for the only two traces $t_0,t_1$ of rank $0$, and for all traces $t,t'$ of rank $r+1$, $t \lord t'$ if and only if there exists $t''$ of rank $r$ such that $S(t'') \in S(t')\setminus S(t)$ and for all $t^*$ of rank $r$ such that $t'' \lord t^*$, $S(t^*) \in S(t)$ if and only if $S(t^*) \in S(t')$. This is defined by analogy with the comparison of two numbers based on their binary representations (assumed to be of the same size). One  is larger than the other if there is a position at which it has a $1$, the other has a $0$, and they are equal on every bit of higher weight. Using this intuition, one can indeed show that $\lord$ is a strict order. 

To express the order~$\lord$ in \hyltl, we use atomic propositions $\ltmark$ and $\gtmark$ and the following formula to ensure that for all traces $t,t'$ in the model, $t \lord t'$ if and only if there exists $j \in \nats$ such that $\ltmark \in t(j)$ and $\gtmark \in t'(j)$:
\[\forall \pi. \forall \pi'. (1_\pi \land 1_{\pi'}) \Rightarrow [\F(\ltmark_\pi \land \gtmark_{\pi'}) \Leftrightarrow \psi_{0,1} \lor (\psi_{equalrank} \land \psi_{order})],\]
where $\psi_{0,1}$ expresses that the value taken by $\pi$ is $t_0$ and the one taken by $\pi'$ is $t_1$, $\psi_{equalrank}$ expresses that the values of $\pi$ and $\pi'$ have the same rank, and \[\psi_{order}  = \exists \tau . \forall \tau'. \psi_{in}(\tau,\pi') \land \neg \psi_{in}(\tau,\pi) \land \F(\ltmark_\tau \land \gtmark_{\tau'}) \Rightarrow (\psi_{in}(\tau,\pi) \Leftrightarrow \psi_{in}(\tau',\pi')).\]

We also require that two traces $t,t'$ of equal rank but incomparable for $\lord$ (i.e. $S(t) = S(t')$) must be equal. This is expressed by the formula \[\forall \pi. \forall \pi'. (1_\pi \land 1_{\pi'}) \Rightarrow \Big[\F(\ltmark_\pi \land \gtmark_{\pi'}) \lor \F(\gtmark_\pi \land \ltmark_{\pi'}) \lor \G \big(\bigwedge_{a \in \ap} a_\pi \Leftrightarrow a_{\pi'}\big)\Big].\]

We then define a successor relation $\suc$ on traces of equal rank. We have $\suc(t_0,t_1)$ and not $\suc(t_1,t_0)$, $\suc(t_1,t_1)$, or $\suc(t_1,t_1)$. For all traces $t,t'$ of rank $r>0$,  $\suc(t,t')$ expresses that there exists a trace $t''$ of rank $r-1$ such that, $S(t'') \in S(t')\setminus S(t)$ and for all $t^* \lord t''$, $S(t^*) \in S(t)\setminus S(t')$. Intuitively, this is analogous to the definition of the successor in binary: In order to add one to a binary number, one has to find the bit of least weight such that this bit is $0$ and all the bits of lower weights are $1$. Then one has to turn this bit to one and all the lower ones to $0$. The following formula is satisfied by $[\pi \rightarrow t, \pi' \rightarrow t']$ if and only if $\suc(t,t')$ holds: \begin{align*}
&    \psi_{\suc}(\pi,\pi') = \exists \tau . \forall \tau'. (1_\pi \land 1_{\pi'}) \land \Big[1_{\tau'}) \Rightarrow\\ 
&1_{\tau} \land \psi_{in}(\tau,\pi') \land \neg \psi_{in}(\tau,\pi) \land (\F(\ltmark_{\tau'} \land \gtmark_\tau) \Rightarrow \psi_{in}(\tau',\pi) \land \neg \psi_{in}(\tau',\pi'))\Big].
\end{align*}

We then ensure that all possible $S(t)$ are generated: We state the existence of a trace $t_0^r$ for each $1 \leq r \leq n$ such that for all $j \in \nats$, $\nimark \notin t_0^r(j)$. The trace $t_0^r$ represents the empty set at rank $r$. The following formula expresses the existence of those traces: \[\exists \tau_0^1\ldots \exists \tau_0^n .\bigwedge_{r=1}^n 1_{\tau_0^r} \land ({\rankprop_r})_{\tau_0^r} \land \G(\neg (\nimark)_{\tau_0^r}).\]

We also state that for all $r$ and for all traces $t$ of rank $r$, if $t$ is not maximal for $\lord$ (i.e., if it does not contain every trace of rank $r-1$), then there exists a trace $t'$ such that $\suc(t,t')$: 
\[\forall \pi. \exists \tau. 1_\pi \Rightarrow \psi_{\suc}(\pi,\tau) \lor \psi_{max}(\pi)\]
where 
\[\psi_{max}(\pi) = \forall \pi'. [1_{\pi'} \land \bigvee_{r=0}^{n-1} {\rankprop_r}_{\pi'} \land {\rankprop_{r+1}}_{\pi}] \Rightarrow \psi_{in}(\pi',\pi).\] 
As a result, we have ensured that all the $S(t)$ were generated. In particular, we have $\twr(n)$ traces of rank $n$, and we have a successor relation $\suc$. We can use these traces as a counter from $0$ to $\twr(n)-1$.

Finally, for convenience, we ensure that $S$ is injective, i.e., for all traces $t,t'$ of type one, if $S(t)=S(t')$ (i.e., if $t$ and $t'$ have the same rank but are incomparable for $\lord$) then $t=t'$. The following formula expresses it: \[\forall \pi. \forall \pi'. \Big[\big(\bigvee_{r=0}^n (\rankprop_r)_\pi \land (\rankprop_r)_{\pi'} \big) \land \neg \F((\ltmark_{\pi} \land \gtmark_{\pi'}) \lor (\ltmark_{\pi'} \land \gtmark_{\pi}))\Big] \Rightarrow \G\big(\bigwedge_{a \in \ap} a_{\pi} \Leftrightarrow a_{\pi'} \big) \]

 We associate to every pair $(t,t')$ of traces of type one and rank $n$ ($t$ for time and $t'$ for space) a unique type two trace representing the content of the tape at the position represented by $t_s$ at the time represented by $t_t$. We use the atomic propositions $\tmark$ and $\smark$ to encode this relation. The following formulas express it: \[\forall \pi_t. \forall \pi_s. \exists \tau. (1_{\pi_t} \land 1_{\pi_s} \land ({\rankprop_n})_{\pi_t} \land ({\rankprop_n})_{\pi_s}) \Rightarrow \F ((\smark)_{\pi_s} \land (\smark)_{\tau}) \land \F ((\tmark)_{\pi_t} \land (\tmark)_{\tau}) \land 2_\tau\]
expresses that to each pair $(t,t')$ of rank $n$ traces we associate a type two trace as explained above, and
\begin{align*}
    &\forall \pi. \forall \pi'. \forall \tau. \Big[\F((\smark)_{\pi} \land (\smark)_{\tau}) \land \F((\smark)_{\pi'} \land (\smark)_{\tau}) \land \F((\tmark)_{\pi} \land (\tmark)_{\tau}) \land \F((\tmark)_{\pi'} \land (\tmark)_{\tau})\Big]\\
    &\Rightarrow \G\big(\bigwedge_{a \in \ap} a_\pi \Leftrightarrow a_{\pi'}\big)
\end{align*}
expresses that this type two trace is unique.

We can now define a formula $\psi_{same time}(\tau,\tau')$ satisfied in the model by $[\tau \rightarrow t, \tau' \rightarrow t']$ if and only if $t$ and $t'$ are of type two and are associated to the same step of the computation: \[\psi_{same time}(\tau,\tau') = \exists \pi_t. 2_\tau \land 2_{\tau'} \land \F((\tmark)_{\pi_t} \land (\tmark)_{\tau}) \land \F((\tmark)_{\pi_t} \land (\tmark)_{\tau'})\]

We also define a formula $\psi_{next time}(\tau,\tau')$ satisfied in the model by $[\tau \rightarrow t, \tau' \rightarrow t']$ if and only if $t$ and $t'$ are of type two and the time associated to $t'$ is the successor of the time associated to $t$:
\[\psi_{next time}(\tau,\tau') = \exists \pi_t. \exists \pi_t'. 2_\tau \land 2_{\tau'} \land \psi_{next}(\pi_t,\pi'_t) \land \F((\tmark)_{\pi_t} \land (\tmark)_{\tau}) \land \F((\tmark)_{\pi'_t} \land (\tmark)_{\tau'})\]
We define $\psi_{next space}$ similarly.

With these formulas at hand, we can now simulate the run of $\M$: The consistency and the transitions of the machine can easily be implemented using the $\suc$ and $\lord$ relations. For instance, the following formula expresses that if at some step of the run the head does not point to a position, then the letter at that position should not change between this step and the next one:
\[\forall \tau. \forall \tau'. \psi_{next time}(\tau,\tau') \land \psi_{same space}(\tau,\tau') \land \neg h_{\tau} \Rightarrow \bigwedge_{a \in \Sigma} (a_\tau \Rightarrow a_{\tau'}). \] 

If the final formula has a model, then $\M$ accepts $w$ , and then we can construct a model satisfying the formula which is finite (as there are finitely many type one and type two traces) and in which all traces are ultimately periodic (as there are finitely many of them, we only need finitely many positions to encode all the relations). Thus, this model can be represented by a Kripke structure of nonelementary size. 

Conversely, if the formula has a model represented by a Kripke structure, then in particular it is satisfiable. Thus, $\M$ accepts $w$ (one can construct an accepting run of $\M$ from a model of the formula). 

Furthermore, even though the construction described in this proof may yield formulas of alternation depth greater than one, we can reduce the depth to one using Lemma \ref{smallAD}: If the final formula has a model, then it has a finite and ultimately periodic one. The construction of Lemma \ref{smallAD} encodes a finite number of new relations over the traces. Hence, the resulting $\forall^* \exists^*$ formula also has a finite and ultimately periodic model, which is therefore also the set of traces of a Kripke structure.
\qed
\end{proof}

Let us conclude by remarking that the satisfiability problem for \hyltlfgone over Kripke structures is different from the general one, i.e., there are satisfiable formulas which are not satisfied by the set of traces of any Kripke structure. Consider for instance the sentence $\forall \pi. \exists \pi'. \G(a_\pi \Rightarrow a_{\pi'}) \land \F(\neg a_\pi \land a_{\pi'})$. It is satisfied by $\set{a}^*\emptyset^\omega$.

Suppose there exists a Kripke structure $\K$ with a set of traces satisfying this sentence. We define inductively an increasing sequence of finite trace prefixes $p_n$ for $n \in \nats$ as $p_0 = \epsilon$ and
$p_{n+1} = p_n\set{a}$ if $p_n\set{a}$ is a prefix of a trace of $\K$, and $p_{n+1} = p_n\emptyset$ otherwise.
 As the set of traces of a Kripke structure is closed, the limit $t$ of this sequence is a trace of $\K$. As $\K$ satisfies the sentence, there exists $t'$ such that for all $j$ if $a \in t(j)$ then $a \in t'(j)$ and there exists $j^*$ such that $a \in t'(j^*)$ and $a \notin t(j^*)$. In particular, there exists a minimal such $j^*$. Then $p_{j^*+1} = p_{j^*}\emptyset$, but $p_{j^*}\set{a}$ is a prefix of $t'$, thereby contradicting the choice of $p_{j^*+1}$, as we prefer to extend by $\set{a}$ instead of $\emptyset$.
Thus, this satisfiable sentence is not satisfiable by the set of traces of a finite Kripke structure.

\section{Conclusion}
\label{conclusion}

We have shown that \hyltl satisfiability can be decidable, either if one restricts the space of models one is interested in to sufficiently simple ones, or if one restricts the alternation and temporal depth of the formulas under consideration.
In particular, we have investigated the formulas of temporal depth one without untils. 
An interesting open problem is to extend the decidability result presented in Theorem~\ref{sat1forall} to formulas with untils.
Also, we claimed no lower bound on the problem solved in Theorem~\ref{sat1forall}. 
We claim there is an $\exps$ lower bound obtained by encoding exponential space Turing machines, but the exact complexity of the problem is left open. 
Another interesting problem left open is the decidability of \hyltlfgone over Kripke structures. 
We have presented a \tower lower bound in Theorem~\ref{SatFX1K}, but it is open whether the problem is indeed decidable. 

In general, restricting the space of models turns out to be more fruitful than to restrict the formulas under consideration, as satisfiability is undecidable for extremely simple formulas (simplicity being measured in alternation depth and temporal depth).
An interesting challenge pertains to finding other measures of simplicity that yield larger decidable fragments.

\bibliographystyle{splncs03}
\bibliography{biblio}

\begin{thebibliography}{10}
\providecommand{\url}[1]{\texttt{#1}}
\providecommand{\urlprefix}{URL }

\bibitem{AgrawalB16}
Agrawal, S., Bonakdarpour, B.: Runtime verification of k-safety hyperproperties
  in {HyperLTL}. In: {CSF} 2016. pp. 239--252. {IEEE} Computer Society (2016)

\bibitem{DBLP:conf/rv/BonakdarpourF16}
Bonakdarpour, B., Finkbeiner, B.: Runtime verification for {HyperLTL}. In:
  Falcone, Y., S{\'{a}}nchez, C. (eds.) {RV} 2016. LNCS, vol. 10012, pp.
  41--45. Springer (2016)

\bibitem{DBLP:conf/post/ClarksonFKMRS14}
Clarkson, M.R., Finkbeiner, B., Koleini, M., Micinski, K.K., Rabe, M.N.,
  S{\'{a}}nchez, C.: Temporal logics for hyperproperties. In: Abadi, M.,
  Kremer, S. (eds.) {POST} 2014. LNCS, vol. 8414, pp. 265--284. Springer (2014)

\bibitem{Clarkson+Schneider/10/Hyperproperties}
Clarkson, M.R., Schneider, F.B.: Hyperproperties. Journal of Computer Security
  18(6),  1157--1210 (2010)

\bibitem{hierarchy}
Coenen, N., Finkbeiner, B., Hahn, C., Hofmann, J.: The hierarchy of hyperlogics
  (2019), to appear at LICS 2019.

\bibitem{hyperliveness}
Coenen, N., Finkbeiner, B., Sanchez, C., Tentrup, L.: Verifying hyperliveness
  (2019), to appear at CAV 2019.

\bibitem{DBLP:journals/iandc/DemriS02}
Demri, S., Schnoebelen, P.: The complexity of propositional linear temporal
  logics in simple cases. Inf. Comput.  174(1),  84--103 (2002)

\bibitem{DBLP:conf/concur/FinkbeinerH16}
Finkbeiner, B., Hahn, C.: Deciding hyperproperties. In: Desharnais, J.,
  Jagadeesan, R. (eds.) {CONCUR} 2016. LIPIcs, vol.~59, pp. 13:1--13:14.
  Schloss Dagstuhl - Leibniz-Zentrum für Informatik (2016)

\bibitem{DBLP:conf/atva/FinkbeinerHH18}
Finkbeiner, B., Hahn, C., Hans, T.: Mghyper: Checking satisfiability of
  hyperltl formulas beyond the {\textbackslash}exists
  {\^{}}*{\textbackslash}forall {\^{}}* {\(\exists\)} {\({_\ast}\)}
  {\(\forall\)} {\({_\ast}\)} fragment. In: Lahiri, S.K., Wang, C. (eds.)
  Automated Technology for Verification and Analysis - 16th International
  Symposium, {ATVA} 2018, Los Angeles, CA, USA, October 7-10, 2018,
  Proceedings. LNCS, vol. 11138, pp. 521--527. Springer (2018),
  \url{https://doi.org/10.1007/978-3-030-01090-4\_31}

\bibitem{DBLP:conf/cav/FinkbeinerHLST18}
Finkbeiner, B., Hahn, C., Lukert, P., Stenger, M., Tentrup, L.: Synthesizing
  reactive systems from hyperproperties. In: Chockler, H., Weissenbacher, G.
  (eds.) {CAV} 2018 (Part {I}). LNCS, vol. 10981, pp. 289--306. Springer (2018)

\bibitem{DBLP:conf/cav/FinkbeinerHS17}
Finkbeiner, B., Hahn, C., Stenger, M.: {EAHyper}: Satisfiability, implication,
  and equivalence checking of hyperproperties. In: Majumdar, R., Kuncak, V.
  (eds.) {CAV} 2017 (Part {II}). LNCS, vol. 10427, pp. 564--570. Springer
  (2017)

\bibitem{DBLP:conf/tacas/FinkbeinerHST18}
Finkbeiner, B., Hahn, C., Stenger, M., Tentrup, L.: {RVHyper}: {A} runtime
  verification tool for temporal hyperproperties. In: Beyer, D., Huisman, M.
  (eds.) {TACAS} 2018 (Part {II}). LNCS, vol. 10806, pp. 194--200. Springer
  (2018)

\bibitem{DBLP:conf/cav/FinkbeinerRS15}
Finkbeiner, B., Rabe, M.N., S{\'{a}}nchez, C.: Algorithms for model checking
  {HyperLTL} and {HyperCTL}{$^*$}. In: Kroening, D., Pasareanu, C.S. (eds.)
  {CAV} 2015 (Part {I}). LNCS, vol. 9206, pp. 30--48. Springer (2015)

\bibitem{DBLP:conf/stacs/Finkbeiner017}
Finkbeiner, B., Zimmermann, M.: The first-order logic of hyperproperties. In:
  Vollmer, H., Vall{\'{e}}e, B. (eds.) {STACS} 2017. LIPIcs, vol.~66, pp.
  30:1--30:14. Schloss Dagstuhl - Leibniz-Zentrum für Informatik (2017)

\bibitem{DBLP:conf/sp/GoguenM82a}
Goguen, J.A., Meseguer, J.: Security policies and security models. In: 1982
  {IEEE} Symposium on Security and Privacy. pp. 11--20. {IEEE} Computer Society
  (1982)

\bibitem{journals/bsl/HalpernHIKVV01}
Halpern, J.Y., Harper, R., Immerman, N., Kolaitis, P.G., Vardi, M.Y., Vianu,
  V.: On the unusual effectiveness of logic in computer science. Bulletin of
  Symbolic Logic  7(2),  213--236 (2001)

\bibitem{DBLP:journals/jacm/Hooper66a}
Hooper, P.K.: The undecidability of the turing machine immortality problem. J.
  Symb. Log.  31(2),  219--234 (1966), \url{https://doi.org/10.2307/2269811}

\bibitem{DBLP:books/lib/HopcroftU69}
Hopcroft, J.E., Ullman, J.D.: Formal languages and their relation to automata.
  Addison-Wesley series in computer science and information processing,
  Addison-Wesley (1969)

\bibitem{Kurshan2018}
Kurshan, R.P.: Transfer of model checking to industrial practice. In: Clarke,
  E.M., Henzinger, T.A., Veith, H., Bloem, R. (eds.) Handbook of Model
  Checking., pp. 763--793. Springer (2018)

\bibitem{DBLP:conf/sp/McCullough88}
McCullough, D.: Noninterference and the composability of security properties.
  In: 1988 {IEEE} Symposium on Security and Privacy. pp. 177--186. {IEEE}
  Computer Society (1988)

\bibitem{DBLP:journals/tse/McCullough90}
McCullough, D.: A hookup theorem for multilevel security. {IEEE} Trans.
  Software Eng.  16(6),  563--568 (1990)

\bibitem{DBLP:conf/sp/McLean94}
McLean, J.: A general theory of composition for trace sets closed under
  selective interleaving functions. In: 1994 {IEEE} Computer Society Symposium
  on Research in Security and Privacy. pp. 79--93. {IEEE} Computer Society
  (1994)

\bibitem{DBLP:journals/jcs/Millen95}
Millen, J.K.: Unwinding forward correctability. Journal of Computer Security
  3(1),  35--54 (1995)

\bibitem{Pnueli/1977/TheTemporalLogicOfPrograms}
Pnueli, A.: {The Temporal Logic of Programs}. In: FOCS 1977. pp. 46--57 (1977)

\bibitem{DBLP:journals/toct/Schmitz16}
Schmitz, S.: Complexity hierarchies beyond elementary. {TOCT}  8(1),  3:1--3:36
  (2016)

\bibitem{Sistla:1985:CPL:3828.3837}
Sistla, A.P., Clarke, E.M.: The complexity of propositional linear temporal
  logics. Journal of the {ACM}  32(3),  733--749 (1985)

\bibitem{DBLP:conf/stoc/StockmeyerM73}
Stockmeyer, L.J., Meyer, A.R.: Word problems requiring exponential time:
  Preliminary report. In: Aho, A.V., Borodin, A., Constable, R.L., Floyd, R.W.,
  Harrison, M.A., Karp, R.M., Strong, H.R. (eds.) STOC 1973. pp. 1--9. {ACM}
  (1973)

\bibitem{DBLP:conf/csfw/ZdancewicM03}
Zdancewic, S., Myers, A.C.: Observational determinism for concurrent program
  security. In: {CSFW} 2003. p.~29. {IEEE} Computer Society (2003)

\end{thebibliography}

\end{document}